\documentclass[a4paper,12pt]{article}
\usepackage{amsmath,amsfonts,amsthm,amssymb,amscd,color}
\usepackage{slashbox, multirow}
\usepackage[dvipdfmx]{graphicx}
\usepackage[format=hang,font=footnotesize]{caption}
\usepackage[dvipdfmx]{graphicx}
\newtheorem{theorem}{Theorem}
\newtheorem{lemma}{Lemma}

\newtheorem{proposition}{Proposition}
\newtheorem{remark}{Remark}
\newtheorem{definition}{Definition}
\newtheorem{example}{Example}

\newcommand{\R}{{\mathbb R}} 

\newcommand{\Z}{{\mathbb Z}}
\newcommand{\N}{{\mathbb N}} 

 \def\im{{\rm i}}
\newcommand{\C}{\mathbb{C}} 
\newcommand{\T}{\mathbb{T}}
 
\newcommand{\stz}{\mathrm{Stz}}

\def\({\left(}
\def\){\right)}
\def\<{\left\langle}
\def\>{\right\rangle}

\newcommand{\Pa}{\partial}

\newcommand{\bra}[1]{ \left<  #1 \right>  }
\newcommand{\braa}[1]{\left(  #1 \right)  }
\newcommand{\brab}[1]{\left\{ #1 \right\} }

\newcommand{\abs}[1]{ \left|  #1 \right|  }
\newcommand{\nr}[1]{  \left\| #1 \right\| }

\newcommand{\wa}{W^\ast}
\newcommand{\wc}{\widetilde{C}_N}

\newcommand{\ve}[2]{
\begin{pmatrix}
#1
\\
#2
\end{pmatrix}
}

\numberwithin{equation}{section}

\title{{\Large {\bf On nonlinear scattering for quantum walks 
}
}}
\author{
{\small 
Masaya Maeda,$^{1}$ 
\footnote{
maeda@math.s.chiba-u.ac.jp
}\quad 
Hironobu Sasaki,$^{1}$ 
\footnote{
sasaki@math.s.chiba-u.ac.jp 
}\quad
Etsuo Segawa,$^{2}$ 
\footnote{
e-segawa@m.tohoku.ac.jp 
}\quad
Akito Suzuki,$^{3}$ 
\footnote{
akito@shinshu-u.ac.jp 
}\quad
Kanako Suzuki,$^{4}$ 
\footnote{
kanako.suzuki.sci2@vc.ibaraki.ac.jp 
}\quad
}\\ 
{\scriptsize $^{1}$ 
Department of Mathematics and Informatics, Faculty of Science, Chiba University, 
}\\
{\scriptsize 
Chiba 263-8522, Japan.
} \\
{\scriptsize $^{2}$ 
Graduate School of Information Sciences, Tohoku University,
}\\
{\scriptsize 
Sendai 980-8579, Japan. 
} \\
{\scriptsize $^3$ 
Division of Mathematics and Physics, Faculty of Engineering, Shinshu University, 
}\\
{\scriptsize 
Nagano 380-8553, Japan 
} \\
{\scriptsize $^4$ 
College of Science, Ibaraki University,
}\\
{\scriptsize 
2-1-1 Bunkyo, Mito 310-8512, Japan
} \\
} 
\date{\empty }
\begin{document}
\maketitle

\par\noindent
\begin{small}
\par\noindent
{\bf Abstract}. 
We study large time behavior of quantum walks (QW) with self-dependent coin.
In particular, we show scattering and derive the reproducing formula for inverse scattering in the weak nonlinear regime. 
The proof is based on space-time estimate of (linear) QW such as Strichartz estimate.
Such argument is standard in the study of nonlinear Schr\"odinger equations but it seems to be the first time to be applied to QW.
We also numerically study the dynamics of QW and observe soliton like solutions. 

\end{small}

\setcounter{equation}{0}

\section{Introduction}
In this paper, we consider the quantum walks (QW) with state dependent quantum coin.
We set  $\mathcal H =  l^2(\Z;\C^2)$ and fix a map $C:\R\times \R\to U(2)$, where $U(2)$ is the set of $2\times 2$ unitary matrices.
We define the (nonlinear) quantum coin $\hat C:\mathcal H\to \mathcal H$ by 
\begin{equation}
\label{defC}
(\hat Cu)(x)=C(|u_1(x)|^2,|u_2(x)|^2) u(x),
\end{equation} where $u=(u_1,u_2)^t\in \mathcal H$.
For $(T_\pm u)(x)=u(x\mp 1)$ and $S=\begin{pmatrix} T_- & 0 \\ 0 & T_+ \end{pmatrix},$
we set
\begin{align}\label{2}
U:=S\hat C:\mathcal H\to\mathcal H.
\end{align}
By definition, $S$ and $\hat C$ preserves the $l^2$ norm, 
and so does $U$.  
Let $u_0 \in \mathcal{H}$
be an initial state for a walker
and satisfy $\|u_0\|_{\mathcal{H}}=1$. 
Then, the state $u(t)$ of the walker at time $t$ is defined by the recursion 
relation 
\begin{align}\label{3}
u(t) = U u(t-1), \quad t=1,2,\cdots,
\end{align}
with $u(0) = u_0$. 
We define a nonlinear evolution operator  $U(t)$ as\[
U(t)u_0 =u(t), \quad t \in \N\cup\{0\}. 
\]
Notice that $u_0 \mapsto u = U(\cdot)u_0$ is a nonlinear map from 
$\mathcal H$ to $l_t^\infty (\N;\mathcal H)$
and $U(t)$ preserves the norm.
As pointed out in \cite{NPR07PRA},  
this nonlinear evolution does not define a quantum system,
but it can be realized in a optical system such as optical Galton board. 
Notice that this is similar to the relation between (linear) Schr\"odinger equation which describes quantum system and nonlinear Schr\"odinger equations which appears in various regions of physics including optics.
Moreover, in a way similar to linear QW, 
we can define
the finding probability $p_t(x)$ of a walker at time $t$ at position $x$ 
through $u(t,x) := (U(t)u_0)(x)$ as
\begin{equation}
\label{find_prob}
p_t(x) = \|u(t,x)\|_{\mathbb{C}^2}^2, 
	\quad (t,x) \in \mathbb{N} \times \mathbb{Z}. 
\end{equation}
Indeed, $p_t$ gives a probability distribution on $\mathbb{Z}$,
because $U(t)$ preserves the norm and
$\sum_{x \in \mathbb{Z}} p_t(x) = \|u_0\|_{\mathcal{H}}^2$.
From these reasons, it is natural to view the system described by $U(t)$
as a nonlinearization of a linear QW
and thus we simply call it 
a nonlinear quantum walk (NLQW). 
We also import terminology from QW and call  $\mathcal{H}$ and vectors in $\mathcal{H}$ the state space and states, respectively.

NLQW have been proposed by several authors for the purpose of investigating more rich dynamics (\cite{NPR07PRA}, \cite{SWH14SR}) or as a simulator of nonlinear Dirac equation \cite{LKN15PRA}.
We now give several examples of NLQW.
The following example by Navarrete, P\'erez and Rold\'an \cite{NPR07PRA} seems to be the first NLQW appeared in the literature.

\begin{example}
Navarrete, P\'erez and Rold\'an  \cite{NPR07PRA} (see also \cite{MDB15PRE}) proposed the following model as a nonlinear generalization of optical Galton board:
\begin{align}\label{4.2}
C(s_1,s_2)=\frac{1}{\sqrt{2}}\begin{pmatrix} 1 & 1\\ 1 & -1 \end{pmatrix} \begin{pmatrix} e^{\im g s_1} & 0 \\ 0 & e^{\im g s_2}\end{pmatrix},
\end{align}
where $g\in \R$.
\end{example}

%


Lee, Kurzy\ifmmode~\acute{n}\else \'{n}\fi{}ski, and Nha \cite{LKN15PRA} proposed two models related to nonlinear Dirac equations.

\begin{example}

For Gross-Neveu model (scaler type interaction)
\begin{align}\label{4.4}
C(s_1,s_2) =\begin{pmatrix} e^{-\im g(s_1-s_2)} & 0 \\ 0 & e^{\im g(s_1-s_2)} \end{pmatrix}R(\theta),
\end{align}
and for Thirring model (vector type interaction)
\begin{align}\label{4.5}
C(s_1,s_2)=e^{\im g (s_1+s_2)}
R(\theta),
\end{align}
where $g,\theta\in\R$ and $R(\theta)=\begin{pmatrix} \cos\theta & -\sin\theta \\ \sin \theta & \cos \theta \end{pmatrix}$.
\end{example}

To construct a NLQW, it suffices to define $C:\R^2\to U(2)$.
Thus, the following NLQW is another natural example.
\begin{example}
Let $\theta:\R^2\to \R$.
We can define an NLQW by $C(s_1,s_2):=R(\theta(s_1,s_2))$.
In particular, setting $\theta(s_1,s_2):=\theta_0+\lambda(gs_1+gs_2)^p$ with $\lambda=\pm1$, we obtain the NLQW with the nonlinear coin
\begin{align}\label{4.5.1}
C(s_1,s_2)=R(\theta_0+\lambda(gs_1+gs_2)^p)=R(\theta_0)R(\lambda(gs_1+gs_2)^p).
\end{align}
We will numerically study the dynamics of this NLQW in section \ref{sec:NS}.
\end{example}

For other models, see \cite{SWH14SR}.
In the following, we restrict our nonlinear coin operator to the following type:
\begin{align}
C(s_1,s_2) = C_0  C_N(gs_1,gs_2), \quad g>0\label{4.6}
\end{align}
where $C_0=\begin{pmatrix} a & b \\-\bar b & \bar a \end{pmatrix}\in U(2)$ ($|a|^2+|b|^2=1$, $0<|a|<1$), $C_N\in C^2([0,\infty)\times [0,\infty);U(2))$ and $C_N(0,0)=I_2$.
Here $I_2=\begin{pmatrix}  1 & 0 \\ 0 & 1 \end{pmatrix}$.
We set
\begin{align}\label{4.7}
U_0=S\hat C_0,\ \text{where}\ \(\hat C_0 u\)(x):=C_0u(x).
\end{align}

The positive parameter $g$ controls the strength of the nonlinearity.
Notice that all models \eqref{4.2}, \eqref{4.4}, \eqref{4.5} \eqref{4.5.1} given above are included in \eqref{4.6}.

In this paper, we view NLQW as a space-time discretized nonlinear Schr\"odinger equation (NLS).
Indeed, we demonstrate that standard estimates such as dispersive estimate and Strichartz estimate hold for QW (Theorem \ref{thm:1}, Lemma \ref{lem:stz}).
These estimates are fundamental tools for the study of NLS and we show that also for NLQW in the weak nonlinear regime, we can show the scattering by parallel argument as the proof of scattering for NLS.

\begin{definition}
We say $U(t)u_0$ scatters if there exists $u_+\in \mathcal H$ s.t.\ $\|U(t)u_0 - U_0^t u_+\|_{l^2}\to 0$ as $t\to \infty$, where $U_0=S C_0$.
\end{definition}

We use $U_{g=1}(t)$ to denote the evolution $U(t)$
that has the nonlinear coin defined in \eqref{4.6} with $g=1$.
We observe that for $v_0 = \sqrt{g} u_0$ with $\|u_0\|_{l^2}=1$, 
\[ U(t)  u_0 = \frac{1}{\sqrt{g}} U_{g=1}(t) v_0. \]
%
Hence, instead of changing $g$, we can fix $g=1$ and vary the norm or $\|u_0\|_{l^2}$. 
Small $\|u_0\|_{l^2}$ will correspond to small $g$.

The main result in this paper is the following.

\begin{theorem}\label{thm:scat}
Assume $C_N\in C^1(\R^2;U(2))$ with $\|C_N(s_1,s_2)-I_2\|_{\C^2\to \C^2}\leq C  (s_1+s_2)^{m}$ and  $\|\partial_{s_j}C_N(s_1,s_2)\|_{\C^2\to \C^2}\leq C  (s_1+s_2)^{m-1}$ for $j=1,2$.
Here, $\|C\|_{\C^2\to \C^2}$ is the operator norm of the matrix $C$.
That is $\|C\|_{\C^2\to \C^2}:=\sup_{v\in \C^2, \|v\|_{\C^2}=1}\|Cv\|_{\C^2}$.
\begin{enumerate}
\item
For the case $m=3$, there exists $\delta>0$ s.t.\ for any $u_0\in l^2$ with $\|u_0\|_{l^2}<\delta$, $U(t)u_0$ scatters.
\item
For the case $m=2$, there exists $\delta>0$ s.t.\ for any $u_0\in l^1$ with $\|u_0\|_{l^1}<\delta$, $U(t)u_0$ scatters.
\end{enumerate}
\end{theorem}

By the scattering, we can prove the weak limit theorem which are extensively studied for QW.
As shown in \cite{Suzuki16QIP}, for a linear evolution $U_{\rm L} = S\hat C_{\rm L}$,
$U_{\rm L}^t u_0$ scatters 
if the position dependent coin $\hat C_{\rm L}$ satisfies 
the short range condition 
such as $C_{\rm L}(x) = C_0 + O(|x|^{-1-\epsilon})$ 
($|x| \to \infty$) 
with $\epsilon > 0$ independent of $x \in \mathbb{Z}$.  
This allows us to establish the weak limit theorem,
{\it i.e.}, the position of the walker scaled by $t^{-1}$ 
converges in law to the a random variable with the Konno distribution. 
This also holds for NLQW.
Let $X_t$ denote the position of a walker described by 
the time evolution $U(t)u_0$,
{\it i.e.}, the distribution of $X_t$ is
\begin{equation}
\label{Xt_dist} 
P(X_t = x) = p_t(x), \quad x \in \mathbb{Z}, 
\end{equation} 
where $p_t(x)$ is defined in \eqref{find_prob}.  
Let $\hat v_0$ be the asymptotic velocity operator
for a homogeneous evolution $U_0$, which is a unique self-adjoint operator
such that $e^{i \xi \hat v_0} = \mbox{s-}\lim_{t \to +\infty} e^{i \xi \hat x_0(t)}$. Here $\hat x_0(t) = U_0^{-t} \hat x U_0$ is 
the Heisenberg operator for the position operator $\hat x$
(see \cite{GJS04PRE,Suzuki16QIP} for more details). 
We use $E_{\hat a}$ to denote the spectral measure 
of a self-adjoint operator $\hat a$. 

\begin{proposition}
\label{p_wlt}
Let $u_0 \in \mathcal{H}$ be a normalized vector. 
If $U(t)u_0$ scatters so that $\|U(t)u_0 - U_0^t u_+\|_{l^2}\to 0$ as $t\to \infty$, then $X_t/t$ converges in law to a random variable $V$,
whose distribution $\mu_V$ is
\[ \mu_V = \|E_{\hat v_0}(\cdot) u_+\|^2. \]
\end{proposition}
A direct calculation yields the fact that 
\begin{equation}
\label{wldist} 
d \|E_{\hat v_0}(v) u_+\|^2 = w(v) f_K(v;|a|), 
\end{equation}
where $f_K$ is the Konno function defined for all $r > 0$ as
\[ f_K(v;r) = \begin{cases} 
	\frac{\sqrt{1-r^2}}{\pi(1-v^2)\sqrt{r^2 -v^2}},
		& |v| < r \\
		0, & |v| \geq 0
		\end{cases} \]
and $w(v)$ the function determined by $u_+$ (see \eqref{defw}
for more details).    
We prove Proposition \ref{p_wlt} and \eqref{wldist} in the appendix.

We next consider the inverse scattering problem, 
which is the problem of identifying unknown nonlinear terms 
under the assumption that 
all of the scattering states are known.
As for inverse scattering problems  
for some nonlinear Schr\"odinger equations and related equations,
there are many papers (see, e.g., 
\cite{%
CarlesGallagher2009,
MorawetzStrauss1973,
Sasaki2012,
Sasaki2015,
SasakiSuzuki2011,
Strauss1974,
Weder1997} 
and references therein).  
Using Theorem \ref{thm:scat} 
and modifying methods in the above papers, 
we obtain a reproducing formula for the nonlinear coin.

For simplicity, we consider the case that $C_N$ can be expressed as $C_N(s_1,s_2)=\tilde C_N(s_1^2,s_2^2)$ with 
\begin{align}\label{A}
\tilde C_N\in C^2(\R^2;U(2))\text{ and }\|\tilde C_N(s_1,s_2)-I_2\|_{\C^2\to \C^2}\leq C  |s_1|+|s_2|.
\end{align}

We define $\delta_{j,x}\in l^1(\Z;\C^2)$ by $\delta_{j,x}(y)= e_j$ if $y=x$ and $\delta_{j,x}(y)=0$ if $y\neq x$ where $e_1=\begin{pmatrix} 1 \\ 0 \end{pmatrix}$ and $e_2=\begin{pmatrix} 0 \\ 1 \end{pmatrix}$.
Further, for $g:\R_+\to \C$, we define $D_\lambda g(\lambda)=\lambda^{-1} \(g(2\lambda)-g(\lambda)\)$.
We define the nonlinear operator $W^*$ by
\begin{align*}
W^* u_0 = u_0 + \sum_{t=0}^\infty U_0^{-t}\(\hat C_N-I_2\) U(t)u_0,
\end{align*}
which is well defined on $\{u_0\in l^1\ |\ \|u_0\|_{l^1}<\delta\}$ by Theorem \ref{thm:scat}.
 
\begin{theorem}[Inverse scattering]\label{thm:invscat}
Assume $(\mathrm{A})$.
Then, we have
\begin{align*}
\left\| 
\begin{pmatrix} \mathcal L_{11}(\lambda)-\lambda D_\lambda \mathcal L_{11}(\lambda)&  D_\lambda \mathcal L_{11}(\lambda)\\
 \mathcal L_{12}(\lambda)-\lambda D_\lambda \mathcal L_{12}(\lambda)& D_\lambda \mathcal L_{12}(\lambda)
\end{pmatrix} -\partial_{1}\tilde C_N(0,0)\right\|_{\C^2\to \C^2}\leq C  \lambda^3,
\end{align*}
and
\begin{align*}
\left\| 
\begin{pmatrix} D_\lambda \mathcal L_{21}(\lambda)& 
\mathcal L_{21}(\lambda) -\lambda D_\lambda \mathcal L_{21}(\lambda)\\ D_\lambda \mathcal L_{22}(\lambda)&  \mathcal L_{22}(\lambda)-\lambda D_\lambda \mathcal L_{22}(\lambda)
\end{pmatrix} -\partial_{2}\tilde C_N(0,0)\right\|_{\C^2\to \C^2}\leq C  \lambda^3,
\end{align*}
where
\begin{align*}
\mathcal L_{1j}&=\lambda^{-10}\<\(U_0^{-1}W^* U_0-W^*\) \(\lambda^2 \delta_{1,0}+\lambda^3 \delta_{2,0}\), \delta_{j,0}\>,\\
\mathcal L_{2j}&=\lambda^{-10}\<\(U_0^{-1}W^* U_0-W^*\) \(\lambda^3 \delta_{1,0}+\lambda^2 \delta_{2,0}\), \delta_{j,0}\>.
\end{align*}
\end{theorem}

Finally, we study NLQW numerically and observe soliton like phenomena for the model \eqref{4.5.1}.
It is shown that there is a traveling wave type solution of NLQW with constant velocity.

The paper is organized as follows.
In section 2, we prove the dispersive estimate and Strichartz estimate for QW with constant coin.
In section 3, we prove Theorem \ref{thm:scat}.
In section 4, we prove Theorem \ref{thm:invscat}.
In section 5, we numerically investigate NLQW with the coin given by \eqref{4.5.1}.
In the appendix, we prove Proposition \ref{p_wlt}.

\section{Dispersive and Strichartz estimates}

We first derive the dispersive estimate for the linear evolution $U_0=S\hat C_0$ by an elementary argument using stationary phase.
We note that this dispersive estimate was first obtained by Sunada and Tate \cite{ST12JFA} in a slightly different form.

We define the (discrete) Fourier transform by
\begin{align}\label{9}
\mathcal F u (\xi):=\sum_{x\in \Z} e^{-\im x \xi}u(x),\quad \xi\in \T:=\R/2\pi \Z,
\end{align}
and 
 the inverse Fourier transform by
\begin{align}\label{10}
\mathcal F^{-1}f(x):=\frac{1}{2\pi}\int_\T e^{\im x \xi}f(\xi)\,d\xi.
\end{align}
Since $S\hat C_0u (x)= P_0 u(x+1)+Q_0u(x+1)$ where $P_0=\begin{pmatrix} a & b \\ 0 & 0 \end{pmatrix}$ and $Q_0 = \begin{pmatrix} 0 & 0 \\ -\bar b & \bar a \end{pmatrix}$, we have
\begin{align*}
\mathcal F(U_0u)(\xi)&=\sum_{x\in \Z}e^{-\im x \xi}\(P_0 u(x+1)+Q_0u(x-1)\)=\sum_{x\in \Z}e^{-\im x \xi}\(e^{\im \xi} P_0 +e^{-\im \xi}Q_0\)u(x)\\&=(e^{\im \xi} P_0 +e^{-\im \xi}Q_0)\mathcal F u(\xi).
\end{align*}
Notice that
\begin{align}\label{11}
\hat U_0(\xi) := e^{\im \xi}P_0+ e^{-\im \xi}Q_0=\begin{pmatrix}e^{\im \xi}a && e^{\im \xi}b\\ - \overline{e^{\im \xi}b} && \overline{e^{\im \xi}a} \end{pmatrix},
\end{align}
is also unitary and the
 eigenvalues are given by
\begin{align}\label{12}
\lambda_\pm(\xi)=w(\xi)\pm \im\sqrt{1-w(\xi)^2}=:e^{\pm\im \tilde p(\xi)},\quad w(\xi)=\mathrm{Re}(e^{\im \xi}a).
\end{align}
Thus diagonalizing $\hat U(\xi)$, we have
\begin{align}\label{12.0}
U_0=\mathcal F^{-1} P^{-1}(\xi)\mathrm{exp}\(\im \begin{pmatrix} \tilde p(\xi) & 0 \\ 0 & - \tilde p(\xi)\end{pmatrix}\)P(\xi)\mathcal F,
\end{align}
where
\begin{align}\label{12.01}
P(\xi)=\frac{1}{|b|^2+|e^{\im \xi}a-\lambda_+(\xi)|^2}
\begin{pmatrix}
-e^{-\im \xi}\bar b & -e^{-\im \xi}a + \lambda_-(\xi) \\ -e^{\im \xi} a + \lambda_+(\xi) & -e^{\im \xi} b
\end{pmatrix}.
\end{align}
We set $a=|a|e^{\im \theta_a}$.
Then, since $0<\tilde p<\pi$, setting
\begin{align}\label{12.1}
p(\xi)=\mathrm{arccos}\(|a|\cos (\xi)\),
\end{align}
we have $\tilde p(\xi)=p(\xi+\theta_a)$.
Differentiating \eqref{12.1}, we obtain
\begin{align}\label{12.2}
p'(\xi) &= \frac{|a|\sin \xi}{\sqrt{1-|a|^2\cos^2\xi}},\\
p''(\xi)&=|a|(1-|a|^2)\frac{\cos\xi}{(1-|a|^2\cos^2\xi)^{3/2}},\label{12.3.2}\\
p'''(\xi)&=-|a|(1-|a|^2)\(1+2|a|^2\cos^2 \xi\)\frac{\sin \xi}{(1-|a|^2\cos^2\xi)^{5/2}}.\label{12.3.3}
\end{align}
By \eqref{12.0}, we have
\begin{align}
\(U_0^tu_0\)(x)&=
\(\(\frac{1}{2\pi}\int_{\T}P^{-1}(\xi)\exp\({\im t \begin{pmatrix} p(\xi+\theta_a)+\frac{\cdot}{t}\xi & 0 \\ 0 & -p(\xi+\theta_a)+\frac{\cdot}{t}\xi \end{pmatrix}}\) P(\xi)\,d\xi\)* u_0\)(x),\label{12.3.0}
\end{align}
where $A*u(x):=\sum_{y\in\Z} A(x-y)u(y)$.
We set the projections $P_\pm$ by
\begin{align*}
P_+:=\mathcal F^{-1}P^{-1}(\xi)\begin{pmatrix} 1 & 0 \\ 0 & 0 \end{pmatrix}P(\xi)\mathcal F,
\end{align*}
and $P_-=1-P_+$.
We define
\begin{align}\label{100}
I_\pm(t,s) = \frac{1}{2\pi}\int_\T e^{\im t\(\pm p(\xi)+s(\xi-\theta_a)\)}Q_\pm(\xi)\,d\xi,
\end{align}
where
\begin{align*}
Q_+=P^{-1}(\xi-\theta_a)\begin{pmatrix} 1 & 0 \\ 0 & 0 \end{pmatrix} P(\xi-\theta_a),\quad Q_-=P^{-1}(\xi-\theta_a)\begin{pmatrix} 0 & 0 \\ 0 & 1 \end{pmatrix} P(\xi-\theta_a).
\end{align*}
Then, we can express the generator by
\begin{align}\label{101}
U_0^t u_0 = \sum_{\pm}U_0^t P_\pm u_0=\sum_{\pm} I_\pm(t,\frac{\cdot}{t})*u_0.
\end{align}

The following is the dispersive estimate for QW.

\begin{theorem}\label{thm:1}
Let $0<|a|<1$.
Then, there exists $C>0$ such that for all $t\geq 1$,
\begin{align}\label{12.3}
\|U_0^t u \|_{l^\infty} \leq C t^{-1/3}\|u\|_{l^1}.
\end{align}
\end{theorem}

\begin{proof}
The proof is similar to Theorem 3 of \cite{SK05N}.
By \eqref{101}, it suffices to show that for $t\geq 1$,
\begin{align}\label{12.3.1}
\sup_{s\in\R}\left\|  I_\pm(t,s) \right\|_{\C^2\to \C^2}:=\max_{1\leq i,j\leq 2}\sup_{s\in\R}|I_{\pm,ij}(t,s)|\leq C t^{-1/3},
\end{align}
where $I_\pm$ are given in \eqref{100} and $I_{\pm,ij}$ are the $(i,j)$ matrix component of $I_{\pm}$.
From \eqref{12.3.2} and \eqref{12.3.3},
\begin{align*}
\(\frac{1+2|a|^2}{1-|a|^2}p''(\xi)\)^2+(p'''(\xi))^2&\geq
\(\frac{1+2|a|^2\cos^2\xi}{1-|a|^2\cos^2\xi}p''(\xi)\)^2+(p'''(\xi))^2\\&=|a|^2(1-|a|^2)^2\frac{\(1+2|a|^2\cos^2\xi\)^2}{(1-|a|^2\cos^2\xi)^{5}}\geq |a|^2(1-|a|^2)^2.
\end{align*}
This implies
$
\min_{\xi \in \T}(|p''(\xi)|, |p'''(\xi)|)>0.
$
Therefore, we can set $\psi_l\in C^\infty$ ($l=1,2$) s.t.\ $\psi_1(\xi)+\psi_2(\xi)=1$, $|p''(\xi)|\geq \delta$ for $\xi \in \mathrm{supp}\psi_1$ and $|p'''(\xi)|\geq \delta$ for $\xi \in \mathrm{supp}\psi_2$.
Now, Theorem \ref{thm:1} follows from Van der Corput lemma:

\begin{lemma}[Van der Corput lemma]\label{lem:1}
Let $\psi\in C^\infty$ and $k\geq 2$ and $|q^{(k)}(\xi)|\geq \delta$ in $\xi \in \mathrm{supp} \psi$.
Then, we have
\begin{align}\label{12.3.5}
\left| \int_{\T}e^{\im t q(\xi)}\psi(\xi) \,d\xi\right|\leq C  (t \delta)^{1/k}
\end{align}
\end{lemma}

\begin{proof}
See \cite{SteinHarmonic}.
\end{proof}

\noindent
From Lemma \ref{lem:1}, we obtain the claim of Theorem \ref{thm:1}.
\end{proof}

As the case of Schr\"odinger equations and discrete Sch\"odinger equaiton (or continuous time QW), we can derive the Strichartz estimate from dispersive estimate.
We define
\begin{align}\label{12.3.6}
\stz = l^\infty_t(\Z_{\geq 0};\l^2_x(\Z))\cap l^6_t(\Z_{\geq 0};l^\infty_x(\Z)),\quad \stz^*=l^1_t(\Z_{\geq 0};l^2_x(\Z))+l^{6/5}_t(\Z_{\geq 0};l^1_x(\Z)),
\end{align}
where $\|u\|_{l^p_t l^q_x}:=\(\sum_{t\geq 0} (\sum_{x\in \Z} |u(t,x)|^q)^{p/q}\)^{1/p}$ and 
\begin{align*}
\|u\|_{\stz}=\max(\|u\|_{l^\infty_t l^2_x}, \|u\|_{l^6_t l^\infty_x}),\quad \|u\|_{\stz^*}=\inf_{u_1+u_2=u}\(\|u_1\|_{l^1_t l^2_x}+\|u_2\|_{l^{6/5}_tl^1_x}\).
\end{align*}
We further define the weak $l^{p}$ space $l^{p,\infty}$ by its norm
\begin{align*}
\|f\|_{l^{p,\infty}}:=\sup_{\gamma>0} \gamma \(\#\{x\in \Z\ |\ |f(x)|>\gamma\}\)^{1/p},
\end{align*}
where $\#$ is the counting measure.
It is well known that $\|f\|_{l^{p,\infty}}\leq \|f\|_{l^p}$ and moreover we have $\|\<\cdot\>^{-1/p}\|_{l^{p,\infty}}<\infty$ ($\<x\>:=(1+|x|^2)^{1/2}$) and the Young's inequality for weak type spaces
\begin{align}\label{12.3.7}
\|f*g\|_{l^{p_0}}\leq C  \|f\|_{l^{p_1,\infty}}\|g\|_{l^{p_2}},
\end{align}
for $1<p_0,p_1,p_2<\infty$ and $1+p_0^{-1}=p_1^{-1}+p_2^{-1}$ (see Theorem 1.4.24 of \cite{GrafakosC}).

By parallel argument for the proof of Strichartz estimate of free Schr\"odinger equation, we have the following discrete Strichatrz estimate.

\begin{lemma}[Strichartz estimate]\label{lem:stz}
We have
\begin{align*}
\|U_0^t u_0\|_{\stz}\leq C\|u_0\|_{l^2},\quad \|\sum_{s=0}^{t} U_0^{t-s}f(s)\|_{\stz}\leq C\|f\|_{\stz^*},
\end{align*}
where $C>0$ is a constant.
\end{lemma}

\begin{proof}
We say $(p,q)$ is admissible if $(p_\theta^{-1},q_\theta^{-1})=(\frac{\theta}{6},\frac{1-\theta}{2})$ for some $\theta\in [0,1]$.
By interpolation between Theorem \ref{thm:1} and $l^2$ conservation, we have $\|U_0^t u_0\|_{l^{q_\theta}}\leq C\<t\>^{-\theta/3}\|u_0\|_{l^{q_\theta'}}$.
Here, $q_\theta' =\frac{q_\theta}{1-q_\theta}$ is the H\"older conjugate.
We first show the dual estimate:
\begin{align}\label{12.3.8}
\|\sum_{s=0}^\infty U_0^{-s}f(s)\|_{l^2}\leq C\|f\|_{\stz^*}
\end{align}
For  admissible $(p_\theta,q_\theta)$, we have
\begin{align}
\|\sum_{s=0}^\infty U_0^{-s} f(s) \|_{l^2}^2 
&=\<\sum_{s=0}^\infty U_0^{-s}f(s),\sum_{s=0}^\infty U_0^{-t}f(t)\>
=\sum_{t=0}\<\sum_{s=0}^\infty U_0^{t-s}f(s),f(t)\>\label{12.3.9}
\\&\leq \sum_{t=0}^\infty \sum_{s=0}^\infty \|U_0^{t-s} f(s) \|_{l^{q_\theta}} \|f(t)\|_{l^{q_\theta'}}
\leq C \sum_{t=0}^\infty \(\sum_{s=0}^\infty \<t-s\>^{-\theta/3} \|f(s)\|_{l^{q_\theta'}}\) \|f(t)\|_{l^{q_\theta'}}\nonumber
\\&\leq C\left\|\<\cdot\>^{-\theta/3}* \|f(\cdot)\|_{l^{q_\theta'}} \right\|_{l^{p_\theta}} \|f\|_{l^{p_\theta'}l^{q_\theta'}}\leq C\|\<\cdot\>^{-\frac{\theta}{3}}\|_{l^{\frac{3}{\theta},\infty}}\|f\|_{l^{p_\theta'}l^{q_\theta'}}^2\leq C\|f\|_{l^{p_\theta'}l^{q_\theta'}}^2,\nonumber
\end{align}
where we have used \eqref{12.3.7} in the third line.
Therefore, we have \eqref{12.3.8}.
Notice that by the same argument we have
\begin{align}\label{12.3.10}
\|\sum_{s=0}^t U_0^{t-s}f(s)\|_{l^\infty l^2}+\|\sum_{s=t}^\infty U_0^{t-s}f(s)\|_{l^\infty l^2}\leq C \|f\|_{l^{p_\theta'}l^{q_\theta'}}.
\end{align}
The first claim follows from a duality argument using \eqref{12.3.8}.
\begin{align*}
\|U_0^t u_0 \|_{\stz}=\sup_{\|f\|_{\stz^*}\leq 1}  \sum_{t=0}^\infty \<U_0^t u_0, f\>=\sup_{\|f\|_{\stz^*}\leq 1}   \< u_0, \sum_{t=0}^\infty U_0^{-t} f\>\leq C\sup_{\|f\|_{\stz^*}\leq 1}\|u_0\|_{l^2} \|f\|_{\stz^*}\leq C\|u_0\|_{l^2}.
\end{align*}
We show the second inequality (inhomogeneous Strichartz).
For admissible $(p_\theta,q_\theta)$, applying the argument of \eqref{12.3.9}, we have
\begin{align}\label{12.3.11}
\|\sum_{s=0}^t U_0^{t-s} f(s)\|_{l^{p_\theta} l^{q_\theta}}\leq C \|\<\cdot\>^{-\frac{\theta}{3}}\|_{l^{\frac{3}{\theta},\infty}}\|f\|_{l^{p_\theta'}l^{q_\theta'}}\leq C \|f\|_{l^{p_\theta'}l^{q_\theta'}}.
\end{align}
Combining \eqref{12.3.10} and \eqref{12.3.11} we have
\begin{align}\label{12.3.12}
\|\sum_{s=0}^t U_0^{t-s} f(s)\|_{\stz}\leq C  \|f\|_{l^{6/5}l^{1}}.
\end{align}
By \eqref{12.3.10}, for admissible pair $(p,q)$,
\begin{align}
\|\sum_{s=0}^t U_0^{t-s}f(s)\|_{l^p l^q}&=\sup_{\|g\|_{l^{p'}l^{q'}}\leq 1} \sum_{t=0}^\infty\< \sum_{s=0}^t U_0^{t-s}f(s),g(t)\>=\sup_{\|g\|_{l^{p'}l^{q'}}\leq 1}\sum_{s=0}^\infty\<  f(s),\sum_{t=s}^\infty U_0^{s-t}g(t)\>\nonumber\\
&\leq \sup_{\|g\|_{l^{p'}l^{q'}}\leq 1}\|f\|_{l^1l^2} \|\sum_{t=s}^\infty U_0^{s-t}g(t) \|_{l^\infty l^2}\leq C \|f\|_{l^1l^2}.
\end{align}
Therefore, by interpolation, we have the conclusion.
\end{proof}

If we only use Strichartz estimate, we can only handle the case $\|C_N(s_1,s_2)\|_{\C^2\to \C^2}\leq C\(  s_1^3+s_2^3\)$.
To lower the power of the nonlinearity, we adapt the idea of Mielke and Patz \cite{MP10AA} (see also \cite{MP12DCDS}).

\begin{theorem}[Improved decay estimate]\label{thm:mp}
Let $0<|a|<1$.
Then, we have
\begin{align}\label{12.4}
\| U_0^t u_0 \|_{l^{4,\infty}}\leq C \<t\>^{-1/4}\|u_0\|_{l^1}.
\end{align}
\end{theorem}

\begin{proof}
By \eqref{101} and Young's inequality for weak type spaces (see Theorem 1.2.13 of \cite{GrafakosC}):
\begin{align*}
\|f*g\|_{l^{p,\infty}}\leq C \|f\|_{l^{p,\infty}} \|g\|_{l^1},
\end{align*}
 it suffices to show that for $t\geq 1$,
\begin{align*}
\max_{1\leq i,j\leq 2}\left\| I_{\pm,ij}(t,\frac{\cdot}{t})\right\|_{l^{4,\infty}}\leq C t^{-1/4},\quad t\geq 1,
\end{align*}
where $I_{\pm}$ is given in \eqref{100} and $I_{\pm,ij}$ are the $(i,j)$ matrix component of $I_{\pm}$.
Further, for $|s|\geq 1$, we can show $|I(t,s)|\leq C (ts)^{-1}$ using integration by parts.
Therefore,
\begin{align*}
\|I_{\pm,ij}(t,\frac{\cdot}{t})\|_{l^4(|x|\geq t)}^4\leq C \sum_{|x|\geq t}|x|^{-4}\leq C t^{-3}.
\end{align*}
Thus, it suffices to show that for each $1\leq i,j\leq 2$, we have
\begin{align}\label{12.6}
\#\left\{x\in \Z_{|\cdot|\leq t} \ |\ \left| I_{\pm,ij}(t,\frac{x}{t})\right|>\gamma\right\}\leq C \gamma^{-4}t^{-1}.
\end{align}
Notice that if $\gamma\leq C  t^{-1/2}$, then \eqref{12.6} is automatically satisfied.
Thus, it suffices to consider the case $\gamma\gg t^{-1/2}$.
For $|s|\leq 1$, we claim
\begin{align}\label{13}
\left| I_{\pm,ij}(t,s)\right|\leq C  t^{-1/2}\(1+ \left|s^2 - |a|^2\right|^{-1/4}\).
\end{align}
\begin{remark}
This statement corresponds to Lemma 3.6 of \cite{MP10AA} and (4.12) of \cite{MP12DCDS}.
However, for the convenience of the readers, we prove it here. 
\end{remark}
If we have \eqref{13}, for $\gamma\gg t^{-1/2}$, we obtain \eqref{12.6} from
\begin{align*}
x\in  \left\{x\in \Z\ |\ \left| I_{\pm,ij}(t,\frac{x}{t})\right|>\gamma\right\} &\Rightarrow \gamma \leq C  t^{-1/2}\(1+ \left|s^2 - |a|^2\right|^{-1/4}\)  \Rightarrow \min_{\pm}(x \pm |a|t) \leq C  \gamma^{-4}t^{-1}.
\end{align*}
Thus, it suffices to show \eqref{13}.
Further, since we already have the global bound \eqref{12.3.1}, it suffices to consider the case
\begin{align}\label{14}
|s^2-|a|^2|\sim \min_{\pm}|s\pm |a| |\gtrsim t^{-2/3}.
\end{align}

For the case $|s|>|a|-t^{-3/2}$, we have $|\pm p'(\xi)+s|\geq |s|-|a|$.
Hence, integration by parts twice, we have
\begin{align*}
|I_{\pm,ij}(t,s)|\leq C  t^{-2}(|s|-|a|)^{-2}\leq C  t^{-1/2}\(1+ \left|s^2 - |a|^2\right|^{-1/4}\).
\end{align*}
For the case $|s|<|a|-t^{-3/2}$, we only consider $I_{+,11}$ and write $Q_{+,11}$ as $Q$ for simplicity.
Without loss of generality, we can assume $s\geq0$.
Set $\delta(s)>0$ so that $p'(-\frac{\pi}{2}\pm \delta(s))+s=0$.

We now fix $\delta_0\ll1$.
Then, if $\delta_0<|a|-|s|$, we have $|p'(\xi)+s|\geq \delta_1$ in $A_{\delta_1}:=\T\setminus \cup_{\pm}(-\frac{\pi}{2}\pm \delta(s)-\delta_1,-\frac{\pi}{2}\pm \delta(s)+\delta_1)$.
Thus, we have
\begin{align}
|I_{\pm,11}(t,s)|&\leq 
\left|\int_A e^{\im t(p(\xi)+s(\xi-\theta_a))}Q(\xi)\,d\xi\right|+\sum_{\pm}\left|\int_{-\frac{\pi}{2}\pm \delta(s)-\delta_1}^{-\frac{\pi}{2}\pm \delta(s)+\delta_1}e^{\im t(p(\xi)+s(\xi-\theta_a))}Q(\xi)\,d\xi\right|\nonumber\\&\leq C  \delta_1+(t \delta_1)^{-1}.\label{15}
\end{align}
Thus, if we take $\delta_1=t^{-1/2}$ in $A(\delta_1)$, we have \eqref{13}.

For the case $t^{-3/2}< |a|-|s|<\delta_0$, we claim 
\begin{align}\label{16}
|p'(\xi)+s|\gtrsim \(|a|^2-|s|^2\)^{1/2}\delta_1+\delta_1^2,\quad\text{for } \xi\in A(\delta_1).
\end{align}
If we have \eqref{16}, take $\delta_1=(|a|^2-|s|^2)^{-1/4}t^{-1/2}$.
Then, by \eqref{14}, we have $\delta_1\leq C  (|a|^2-s^2)^{1/2}$ and we can replace r.h.s.\ of \eqref{16} by $\(|a|^2-|s|^2\)^{1/2}\delta_1$.
Therefore, estimate of \eqref{15} with $(t \delta_1)^{-1}$ replaced by $(t(|a|^2-|s|^2)^{1/2}\delta_1)^{-1}$ will give us \eqref{13}.

Finally, we show \eqref{16}.
First, notice that as $s\to |a|$, $\delta(s)\to 0$.
Therefore, if $0<|a|-|s|<\delta_0\ll1$, we have $0<\delta(s)\ll1$ and
\begin{align*}
p'(-\pi/2 + \delta(s))+s= -|a|+s +\frac 1 2|a|(1-|a|^2)\delta(s)^2+o(\delta(s)^2),
\end{align*}
where we have used $p''(-\pi/2)=0$ and $p'''(-\pi/2)=|a|(1-|a|^2)$.
This gives us 
\begin{align}
\delta(s)= \sqrt{\frac{2(|a|-s)}{|a|(1-|a|^2)}}+o(\sqrt{|a|-s})
\end{align}
Since $\inf_{\xi\in A(\delta_1)}|p'(\xi)+s|$ is given by the minimum of $|p'\(-\pi/2 \pm(\delta(s)+\delta_1)\)+s|$, we have
\begin{align*}
|p'(-\pi/2 \pm(\delta(s)+\delta_1))+s|&=p'(-\pi/2 \pm \delta(s))+p''(-\pi/2 \pm \delta(s))(\pm \delta_1)+\frac 1 2 p'''(-\pi/2 \pm \delta(s))\delta_1^2\\&=
 \frac{1}{2}|a|(1-|a|^2)\(\delta(s)\delta_1+\delta_1^2\)+o(\delta_1^2). 
\end{align*}
Therefore, we have \eqref{16}.
\end{proof}

\section{Scattering}

We now prove Theorem \ref{thm:scat}

\begin{proof}[Proof of Theorem \ref{thm:scat} 1.]
We first estimate the Strichartz norm.
Set $\Phi:l^\infty l^2\to l^\infty l^2$ by
\begin{align*}
\Phi(u)(t)=U_0^t u_0+\sum_{s=0}^{t-1}U_0^{t-s} \(\hat C_N-I_2\)u(s).
\end{align*}
Notice that $U(t)u_0$ is the unique solution of \eqref{3} if and only if it is a fixed point of $\Phi$.
We show that if $\delta:=\|u_0\|_{l^2}\ll1$, then $\Phi$ has an fixed point.
Indeed, by lemma \ref{lem:stz}
\begin{align*}
\|\Phi(0)\|_{\stz}\leq C  \|u_0\|_{l^2},
\end{align*}
and
\begin{align*}
&\|\Phi(u)-\Phi(v)\|_{\stz}\leq C  \| U_0(\hat C-I_2)u-U_0(\hat C-I_2)v \|_{\stz^*}\leq C  \|(\hat C-I_2)u-(\hat C-I_2)v\|_{l^1l^2}\\&\leq C  \sum_{t\geq 0} \(\sum_{x\in \Z} \| C_N(|u_1(t,x)|^2,|u_2(t,x)|^2)-I_2\|_{\C^2\to\C^2}^2 \| u(t,x)-v(t,x)\|_{\C^2}^2\)^{1/2} \\&\quad+ C\sum_{t\geq 0} \(\sum_{x\in \Z} \|C_N(|u_1(t,x)|^2,|u_2(t,x)|^2)-C_N(|u_1(t,x)|^2,|u_2(t,x)|^2)\|_{\C^2\to\C^2}^2 \| v(t,x)\|_{\C^2}^2\)^{1/2}\\&\leq C 
\left\| \(\|u\|_{\C^2}^{6}+\|v\|_{\C^2}^{6}\) \| u-v\|_{\C^2}\right\|_{l^1l^2}\leq C  \left\| \|u\|_{\C^2}^6+\|v\|_{\C^2}^6 \right\|_{l^1 l^\infty} \|u-v \|_{l^\infty l^2}
\\&
\leq C  
\(\|u\|_{\stz}^6+\|v\|_{\stz}^6\)\|u-v\|_{\stz}.
\end{align*}
Thus, if we set $\mathcal B:=\{ u\in \stz\ |\ \|u\|_{\stz}\leq 2 \|u_0\|_{l^2}=2 \delta\}$,
we see that $\Phi:\mathcal B\to \mathcal B$ is a contraction mapping, provided $\delta>0$ sufficiently small.
Therefore, there exists a unique $u$ s.t. $\Phi(u)=u$, which is actually $U(t)u_0$
Further, since
\begin{align*}
U_0^{-t}U(t) u_0 = u_0 +\sum_{s=0}^{t-1}U_0^{-s}\(\hat C_N-I_2\)u(s),
\end{align*}
if we can show the right hand side is Cauchy in $\mathcal H$, we have the conclusion.
However, we already have the Strichartz bound $\|u\|_{\stz}\ll1$, so by the dual Strichartz estimate \eqref{12.3.8}, we have
\begin{align*}
\|\sum_{s=0}^{\infty} U_0^{-s}\(\hat C_N-I_2\)u(s)\|_{l^2}\leq C  \|u\|_{\stz}^7<\infty,
\end{align*}
and we conclude
\begin{align*}
\|\sum_{s=t_1}^{t_2}U_0^{-s}\(\hat C_N-I_2\)u(s)\|_{l^2}\to 0,\quad t_1\to \infty.
\end{align*}
Thus,
\begin{align*}
u_0 +\sum_{s=0}^{t-1}U_0^{-s-1}\(\hat C_N-I_2\)u(s)\to u_+,\quad t\to \infty.
\end{align*}
Therefore, we have the conclusion.
\end{proof}

To show Theorem \ref{thm:scat} 2., we first show the decay of $l^5$ norm.
Notice that $\<t\>^{-4/15}$ is the same rate for the decreasing of linear solution, which is obtained by interpolation between $l^\infty$-$l^1$ estimate (Theorem \ref{thm:1}) and $l^{4,\infty}$-$l^1$ estimate (Theorem \ref{thm:mp}).
\begin{lemma}\label{lem:qdec}
Under the assumption of Theorem \ref{thm:scat} 2., there exists $\delta>0$ such that if $\|u(0)\|_{l^1}\leq \delta$, we have $\|U(t) u_0\|_{l^5} \leq C  \<t\>^{-4/15}\|u_0\|_{l^1}$.
\end{lemma}

\begin{proof}
We prove by induction.
Suppose $\|u(0)\|_{l^1}= \delta$.
Then, by Theorem \ref{thm:mp}, there exists $c_1>0$ s.t.\ $\|U_0^t u(0)\|_{l^5}\leq c_1 \<t\>^{-4/15} \|u(0)\|_{l^1}$.
We assume that for $0\leq s\leq t-1$, we have $\|u(s)\|_{l^5}\leq 2c_1 \delta\<s\>^{-4/15}$.
Then, we have
\begin{align*}
\<t\>^{4/15}\|u(t)\|_{l^5}&\leq c_1 \delta + \sum_{s=0}^{t-1}\<t\>^{4/15}c_1\<t-s-1\>^{-4/15}\|u(s)^5\|_{l^1}\\&\leq
c_1 \delta +  (2 c_1 \delta)^5 c_1\sum_{s=0}^{t-1}\<t\>^{4/15}\<t-s-1\>^{-4/15}\<s\>^{-4/3}.
\end{align*}
Notice that 
\begin{align*}
\sum_{s=0}^{t-1}\<t\>^{4/15}\<t-s-1\>^{-4/15}\<s\>^{-3/4}\leq c_2
\end{align*}
with some absolute constant $c_2>0$.
Indeed,
\begin{align*}
\sum_{s=0}^{t-1}\<t-s-1\>^{-4/15}\<s\>^{-3/4}&\sim \int_0^{t}\frac{1}{(1+t-s)^{4/15}}\frac{1}{(1+s)^{4/3}}
\\&\leq C  \<t\>^{-4/15}\int_0^{t/2}\<s\>^{-4/3}+\<t\>^{-4/3}\int_{t/2}^t (1+t-s)^{-4/15}\,ds
\\&\leq C  \<t\>^{-4/15}+\<t\>^{-4/3}\<t\>^{-4/15+1}\leq C   \<t\>^{-4/15}.
\end{align*}
Thus, if we take $0<\delta$ to satisfy $\delta<(2c_1 c_2^{1/5})^{-1}$, we have
\begin{align*}
c_1 \delta(1 +  32c_1^5\delta^5	c_2)< 2c_1 \delta.
\end{align*}
Therefore, we have the conclusion.
\end{proof}

We can show scattering by using decay.

\begin{proof}[Proof of Theorem \ref{thm:scat} 2.]
Let $\|u(0)\|_{l^1}\leq \delta$, where $\delta>0$ given by Lemma \ref{lem:qdec}. Then we have $\|u(t)\|_{l^5}\leq C  \<t\>^{-4/15}$.
Thus, $\|u\|_{l^{24/5}l^5(t,\infty)}\to 0$ as $t\to \infty$.
Therefore, we have
\begin{align*}
\|u\|_{\stz(T,\infty)}&\leq C  \|u(T)\|_{l^2}+C\|\sum_{s=T}^{t-1} U_0^{t-s}(\hat C_N-I_2)u(s)\|_{\stz}\\& \leq C  \|u(T)\|_{l^2} + C\| U_0(\hat C-I_2)u\|_{l^{6/5}l^1}\\&
\leq C  \|u(T)\|_{l^2} + C\|u\|_{l^{24/5}l^8(T,\infty)}^4 \|u\|_{l^\infty l^2}\\&
\leq C  \|u(T)\|_{l^2} + C\|u\|_{l^{24/5}l^5(T,\infty)}^4 \|u\|_{\stz}
\end{align*}
Therefore, we see that $\stz$ norm is finite.
Since we can bound $\|\sum_{s=0}^\infty U_0^{-s}(\hat C_N-I_2)u(s)\|_{l^2}$ by the same estimate, we have the conclusion.
\end{proof}

\section{Inverse scattering}

In the following, we always assume \eqref{A}.
\begin{lemma}\label{lem:2:thm:is:1}
Let $\delta>0$ sufficiently small.
Then, for any $u_0 \in  l ^1$ 
with $\|u_0\|_{ l ^1}\le \delta$,
\begin{align}\label{ineq:1:lem:2:thm:is:1}
\| U(t)u_0 - U_0^t u_0 \|_{ l ^\infty  l ^2}
\leq C  \|u_0\|_{ l ^1}^5
\end{align} 
and 
\begin{align}\label{ineq:2:lem:2:thm:is:1}
\| \(\hat C_N-I_2\) U(t)u_0 - \(\hat C_N-I_2\) U_0^t u_0 \|_{ l ^{6/5}  l ^1}
\leq C  \|u_0\|_{ l ^1}^9.
\end{align} 
\end{lemma}

\begin{proof}
We have 
\begin{align*}
\| U(t)u_0 - U_0^t u_0 \|_{ l ^\infty  l ^2}
=&
\| \sum_{s=0}^{t-1}U_0^{t-s}\(\hat C_N-I_2\) U(s)u_0 \|_{ l ^\infty  l ^2}
\leq C 
\| \(\hat C_N-I_2\) U(t)u_0 \|_{ l ^{6/5} l ^1}
\\
\leq C 
&
\| | U(t)u_0 |^5 \|_{ l ^{6/5} l ^1}
\leq C 
\| U(t)u_0 \|_{ l ^{24/5} l ^5}^4 
\| U(t)u_0 \|_{ l ^\infty  l ^2}
\leq C 
\| u_0 \|_{ l ^1}^5,
\end{align*} 
where we have used Strichartz estimate (Lemma \ref{lem:stz}), assumption \eqref{A} and H\"older inequality.
Thus, we have \eqref{ineq:1:lem:2:thm:is:1}. 
\\

Let $u_1=\ve{u_{11}}{u_{12}} \in  l ^1$ 
and $u_2=\ve{u_{21}}{u_{22}} \in  l ^1$.
Then we have for any $x\in \Z$, 
\begin{align*}
&
\(\hat C_N-I_2\) u_1(x) - \(\hat C_N-I_2\) u_2(x)
\\
&=
\braa{\wc\braa{|u_{11}(x)|^4,|u_{12}(x)|^4}-I_2}u_1
-
\braa{\wc\braa{|u_{21}(x)|^4,|u_{22}(x)|^4}-I_2}u_2
\\
&=
\braa{\wc\braa{|u_{11}(x)|^4,|u_{12}(x)|^4}-I_2}(u_1(x)-u_2(x))
\\
&\quad +
\brab{
\braa{\wc\braa{|u_{11}(x)|^4,|u_{12}(x)|^4}-I_2}
-
\braa{\wc\braa{|u_{21}(x)|^4,|u_{22}(x)|^4}-I_2}
}u_2(x).
\end{align*} 
Using (A1), 
we obtain for any $x\in \Z$,
\begin{align*}
&
\abs{
\(\hat C_N-I_2\) u_1(x) - \(\hat C_N-I_2\) u_2(x)
}_{\C^2}
\\
&\leq C 
\braa{
|u_{11}(x)|^4+|u_{12}(x)|^4
}
\abs{ u_1(x)-u_2(x) }_{\C^2}
\\
&\quad +
\braa{
|u_{11}(x)|^4-|u_{21}(x)|^4+|u_{21}(x)|^4-|u_{22}(x)|^4
}
|u_2(x)|_{\C^2}
\\
&\leq C 
\braa{
|u_1(x)|_{\C^2} + |u_2(x)|_{\C^2}
}^4
\abs{ u_1(x)-u_2(x) }_{\C^2}.
\end{align*} 
Hence it follows that 
\begin{align*}
&
\left\|
\(\hat C_N-I_2\) U(t)u_0 - \(\hat C_N-I_2\) U_0^t u_0 
\right\|_{ l ^{6/5} l ^1}
\\&\leq C 
\braa{
  \nr{U(t)  u_0}_{ l ^{24/5} l ^5} 
+ \nr{U_0^t u_0}_{ l ^{24/5} l ^5}
}^4
  \nr{ U(t)u_0-U_0^t u_0 }_{ l ^\infty  l ^2}
\leq C 
\nr{u_0}_{ l ^1}^9,
\end{align*} 
where we have used \eqref{ineq:1:lem:2:thm:is:1} and Lemma \ref{lem:qdec} in the last line.
This completes the proof of  \eqref{ineq:2:lem:2:thm:is:1}.
\end{proof}

\begin{proof}[Proof of Theorem \ref{thm:invscat}]
Let $u_0\in  l ^1$ with $\nr{u_0}_{ l ^1}\le \delta$ 
and let $v_0 \in  l ^1$.
Since 
\begin{align*}
\wa u_0 = u_0 + \sum_{t=0}^\infty U_0^{-t}\(\hat C_N-I_2\) U(t)u_0,
\end{align*} 
we have 
\begin{align*}
\bra{ \wa u_0 - u_0 , v_0}
=
\sum_{t=0}^\infty \bra{ \(\hat C_N-I_2\) U(t)u_0 ,U_0^t v_0}.
\end{align*} 
By Lemmas \ref{lem:stz} and \ref{lem:2:thm:is:1}, 
we obtain 
\begin{align*}
&
\abs{
\bra{ \wa u_0 - u_0 , v_0}
-
\sum_{t=0}^\infty \bra{ \(\hat C_N-I_2\) U_0^t u_0 ,U_0^t v_0}
}
\\
&
\leq C 
\nr{ \(\hat C_N-I_2\) U(t)u_0-\(\hat C_N-I_2\) U_0^t u_0 }_{ l ^{6/5} l ^1}
\nr{ U_0^t v_0 }_{ l ^{6  } l ^\infty}
\\
&\leq C 
\nr{u_0}_{ l ^1}^9 
\nr{v_0}_{ l ^2}.
\end{align*}
Replacing $u_0$ and $v_0$ 
by $U_0u_0$ and $U_0v_0$, respectively, 
we have
\begin{align*}
&
\abs{
\bra{ U_0^{-1}\wa U_0 u_0 -u_0 , v_0}
-
\sum_{t=1}^\infty \bra{ \(\hat C_N-I_2\) U_0^t u_0 ,U_0^t v_0}
}
\\
&=
\abs{
\bra{ \wa U_0 u_0 - U_0 u_0 , U_0 v_0}
-
\sum_{t=0}^\infty \bra{ \(\hat C_N-I_2\) U_0^{t+1} u_0 ,U_0^{t+1} v_0}
}
\\
&\leq C 
\nr{U_0u_0}_{ l ^1}^9 
\nr{U_0v_0}_{ l ^2}
\leq C 
\nr{u_0}_{ l ^1}^9 
\nr{v_0}_{ l ^2},
\end{align*} 
which implies that 
\begin{align*}
&
\abs{
\bra{ \braa{U_0^{-1}\wa U_0 -\wa}u_0 , v_0}
-
\bra{ \(\hat C_N-I_2\) u_0 ,v_0}
}
\leq C 
\nr{u_0}_{ l ^1}^9 
\nr{v_0}_{ l ^2}.
\end{align*} 
In particular, for $k=1,2$ and $l=1,2$
we see that 
\begin{align}
&
\left|
\< \braa{U_0^{-1}\wa U_0 -\wa}
 \(\lambda^{1+k}\delta_{1,0}+\lambda^{4-k}\delta_{2,0}\),
\delta_{l,0}\>
-
\< \(\hat C_N-I_2\) \(\lambda^{1+k}\delta_{1,0}+\lambda^{4-k}\delta_{2,0}\),
\delta_{l,0}\>
\right|
\nonumber\\
&\leq C 
\left\|\lambda^{1+k}\delta_{1,0}+\lambda^{4-k}\delta_{2,0}\right\|_{ l ^1}^9 
\|\delta_{l,0}\|_{ l ^2}
\leq C  
\lambda^{18}\label{ineq:1:proof}
\end{align} 
for any $k=1,2$ and for any $\lambda>0$ sufficiently small.
By the Taylor theorem and \eqref{A}, 
we have 
\begin{align*}
&\< \(\hat C_N-I_2\) \(\lambda^{1+k}\delta_{1,0}+\lambda^{4-k}\delta_{2,0}\),
\delta_{l,0}\>
=
\bra{ \braa{\wc(\lambda^{4(1+k)},\lambda^{4(4-k)})-I_2}
\ve{\lambda^{1+k}}{\lambda^{4-k}},
\delta_{l,0}}_{\C^2}
\\
&=
\bra{ \Pa_1\wc(0,0)\lambda^{4(1+k)} \ve{\lambda^{1+k}}{\lambda^{4-k}},
\delta_{l,0}}_{\C^2}
+
\bra{ \Pa_2\wc(0,0)\lambda^{4(4-k)} \ve{\lambda^{1+k}}{\lambda^{4-k}},
\delta_{l,0}}_{\C^2} +O(\lambda^{18}).
\end{align*} 
We see from (\ref{ineq:1:proof}) that 
\begin{align*}
\mathcal L_{1j}(\lambda)
&=
\lambda^{-10}
\< \(\hat C_N-I_2\) \(\lambda^{2}\delta_{1,0}+\lambda^{3}\delta_{2,0}\),
\delta_{j,0}\>
+O(\lambda^8)
\\
&=
\bra{ \Pa_1\wc(0,0)\ve{1}{\lambda},
\delta_{j,0}}_{\C^2}
+O(\lambda^4)
\\
&=
\(\Pa_1\wc(0,0)\)_{j1}+\lambda\(\Pa_1\wc(0,0)\)_{j2}
+O(\lambda^4)
\quad (\lambda\to+0)
\end{align*}
and
\begin{align*}
\mathcal L_{2j}(\lambda)
&=
\lambda^{-10}
\< \(\hat C_N-I_2\) \(\lambda^{3}\delta_{1,0}+\lambda^{2}\delta_{2,0}\),
\delta_{j,0}\>
+O(\lambda^8)
\\
&=
\bra{ \Pa_2\wc(0,0)\ve{\lambda}{1},
\delta_{j,0}}_{\C^2}
+O(\lambda^4)
\\
&=
\(\Pa_2\wc(0,0)\)_{j2}+\lambda\(\Pa_2\wc(0,0)\)_{j1}
+O(\lambda^4)
\quad (\lambda\to+0)
\end{align*}
for $j=1,2$.
Here, 
we have defined $\(\partial_k \tilde C_N(0,0)\)_{ij}=\<\partial_k \tilde C_N(0,0) e_j,e_i\>_{\C^2}$.

Hence we have
\begin{align*}
\(\Pa_1\wc(0,0)\)_{j2}&=D_\lambda \mathcal L_{j1}+O(\lambda^3), \quad \(\Pa_2\wc(0,0)\)_{j1}=D_\lambda \mathcal L_{j2}+O(\lambda^3),\\
\(\Pa_1\wc(0,0)\)_{j1}&=\mathcal L_{1j}(\lambda)-\lambda D_\lambda \mathcal L_{1j}+O(\lambda^4), \quad \(\Pa_2\wc(0,0)\)_{j2}=\mathcal L_{2j}-\lambda D_\lambda \mathcal L_{2j}+O(\lambda^4),
\end{align*}
which completes the proof.
\end{proof}

\section{Numerical simulation for NLQW and soliton like behavior}\label{sec:NS}
\begin{figure}[t]
\begin{tabular}{cc}
\begin{minipage}{0.45\hsize}
(a)
\vspace{-0.5cm}
\begin{center}
\includegraphics[width=6.0cm]{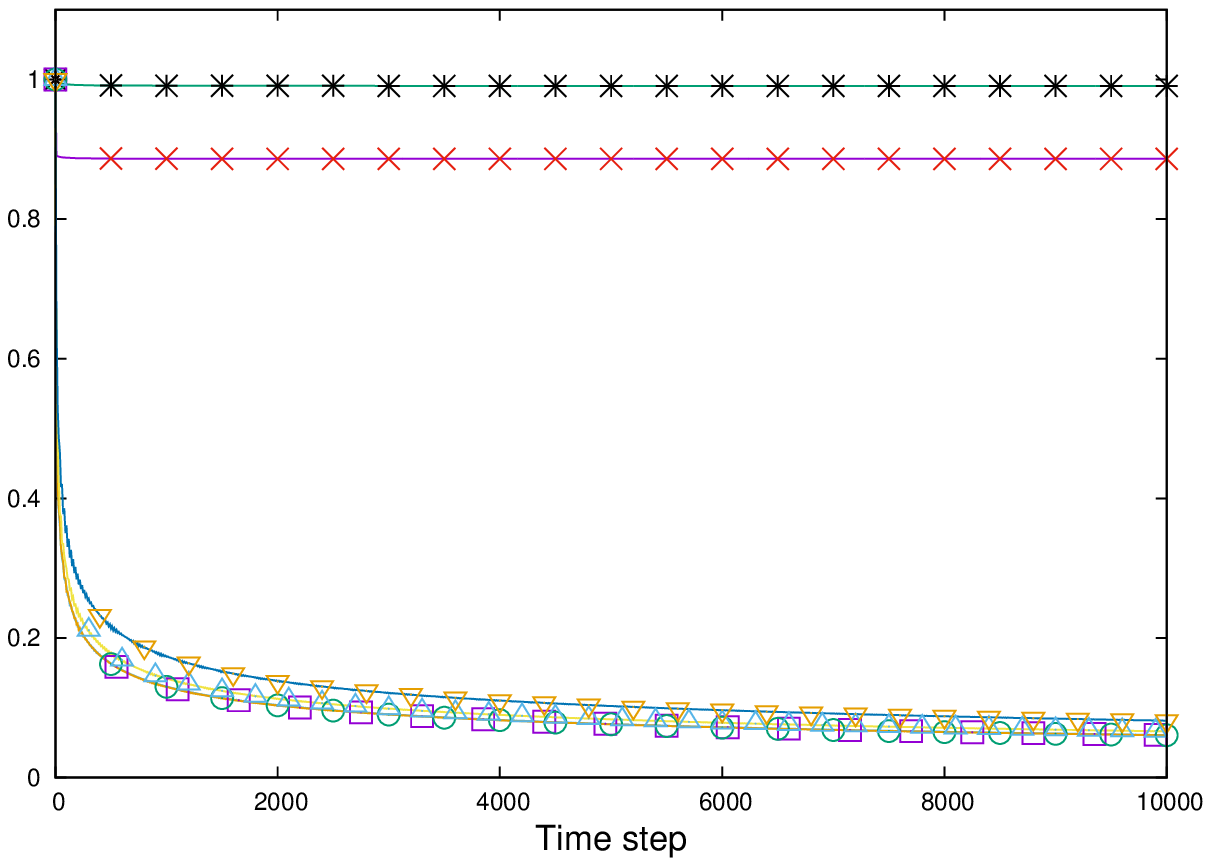} 
\end{center}
\end{minipage}
\hspace{6mm}
\begin{minipage}{0.45\hsize}
(b)
\vspace{-0.5cm}
\begin{center}
\includegraphics[width=6.0cm]{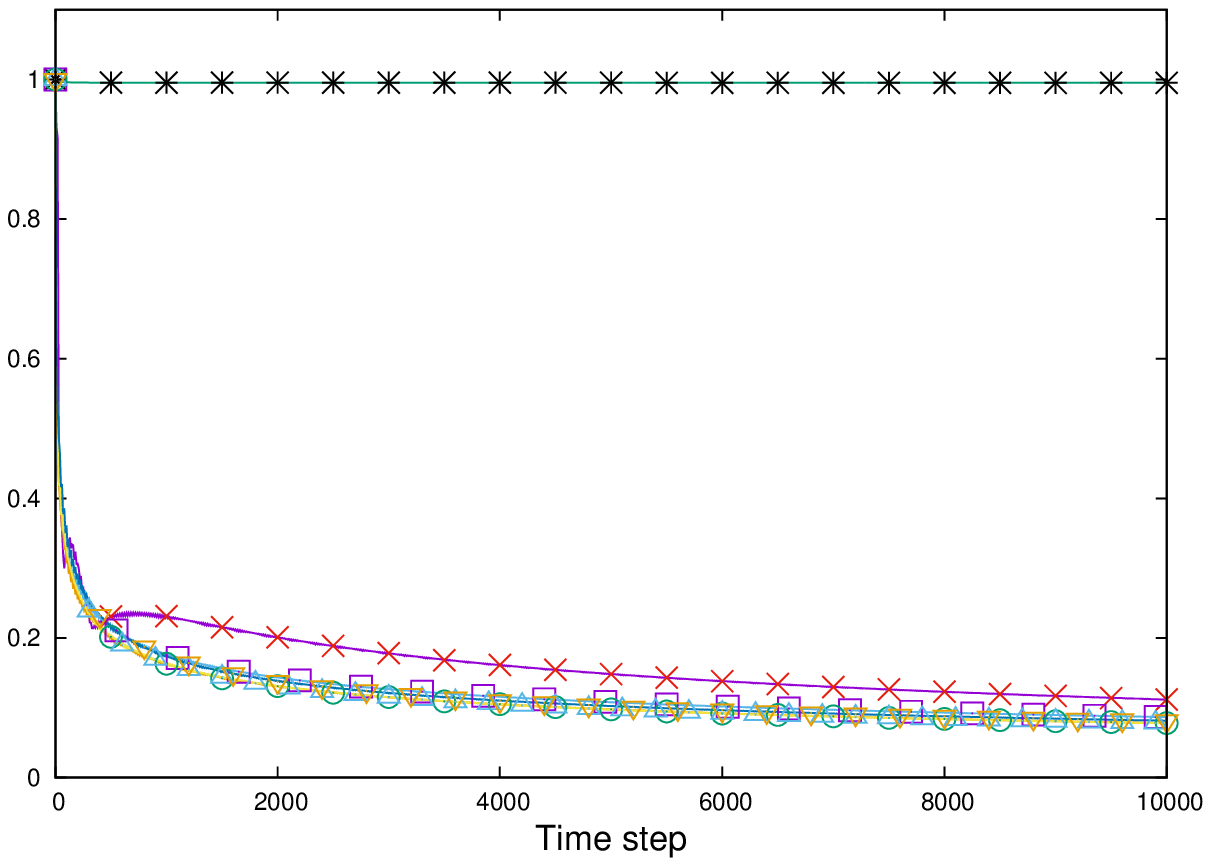}
\end{center}
\end{minipage}\vspace{0.3cm}\\
\begin{minipage}{0.45\hsize}
(c)
\vspace{-0.5cm}
\begin{center}
\includegraphics[width=6.0cm]{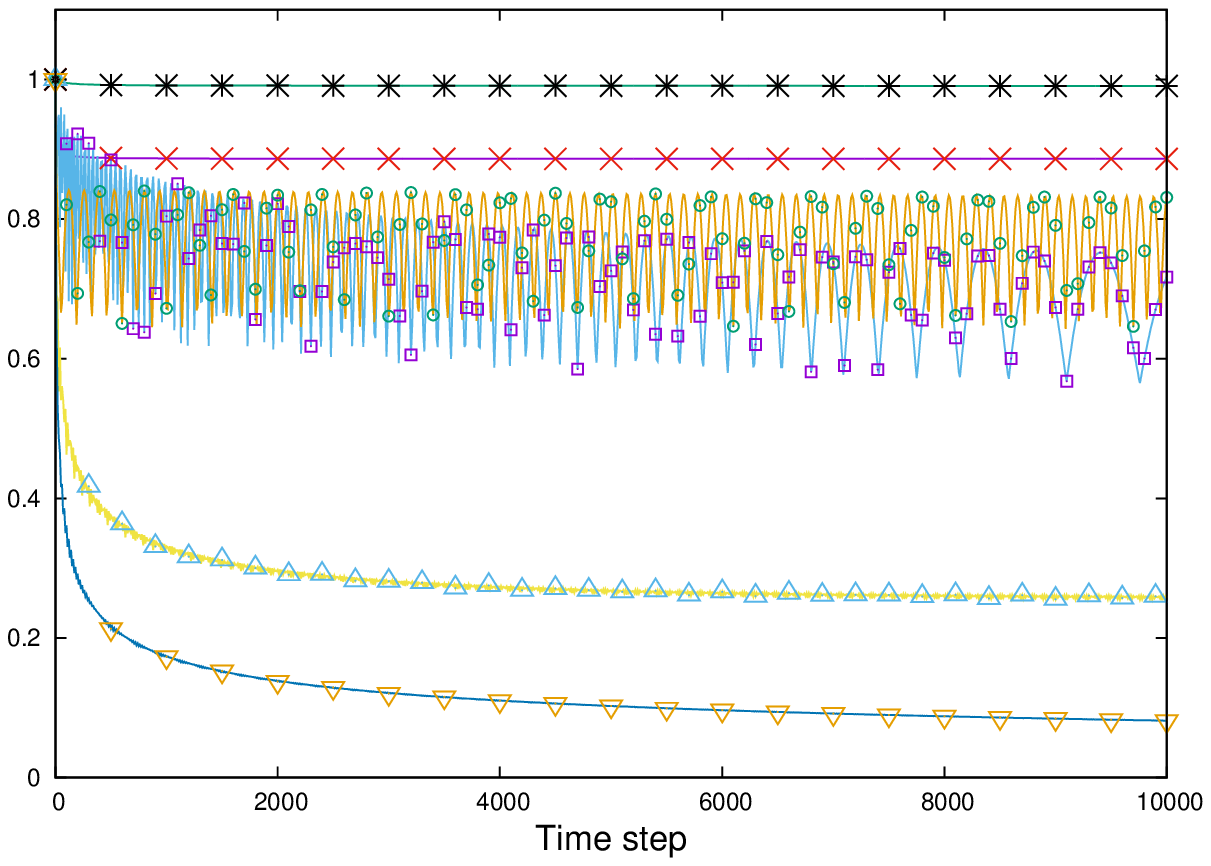}
\end{center}
\end{minipage}
\hspace{6mm}
\begin{minipage}{0.45\hsize}
(d)
\vspace{-0.5cm}
\begin{center}
\includegraphics[width=6.0cm]{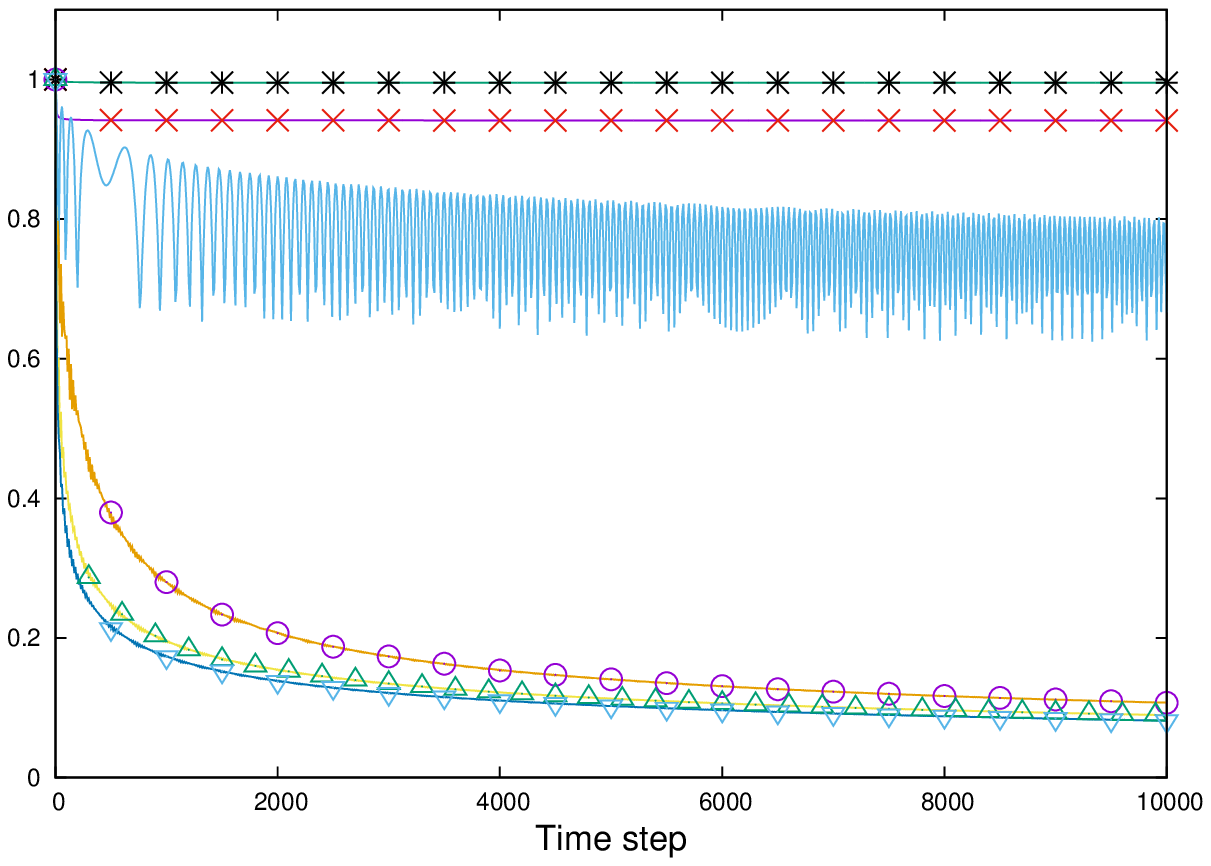}
\end{center}
\end{minipage}
\end{tabular}
\caption{Figures (a)--(d) show the behavior of $\|u\|_{l^\infty}$, in which the vertical axis is $\|u\|_{l^\infty}$ and the horizontal axis is time step. Figures (a) and (b) are the case $g > 0$ for $p = 1$ and $p = 2$, respectively. In both figures, there are $g = 1\ (\times)$, $g = 0.8\ (\ast)$, $g = 0.6 \ (\Box)$, $g = 0.4 \ (\bigcirc)$, $g = 0.2\ (\triangle)$, $g = 0.0\ (\triangledown)$. Figures (c) and (d) are the case $g < 0$ for $p = 1$ and $p = 2$, respectively. In both figures, there are $g = -1\ (\times)$, $g = -0.8\ (\ast)$, $g = -0.6 \ (\Box)$, $g = -0.4 \ (\bigcirc)$, $g = -0.2\ (\triangle)$, $g = 0.0\ (\triangledown)$. In (c), $g = -0.6$ and $g = -0.4$ show oscillations, while only the case $g = -0.6$ oscillates in (d) (it is plotted without $\Box$). }
\label{sup-norm-u}
\end{figure}
\begin{figure}[t]
\begin{minipage}{0.45\hsize}
(a)
\vspace{-0.5cm}
\begin{center}
\includegraphics[width=6.0cm]{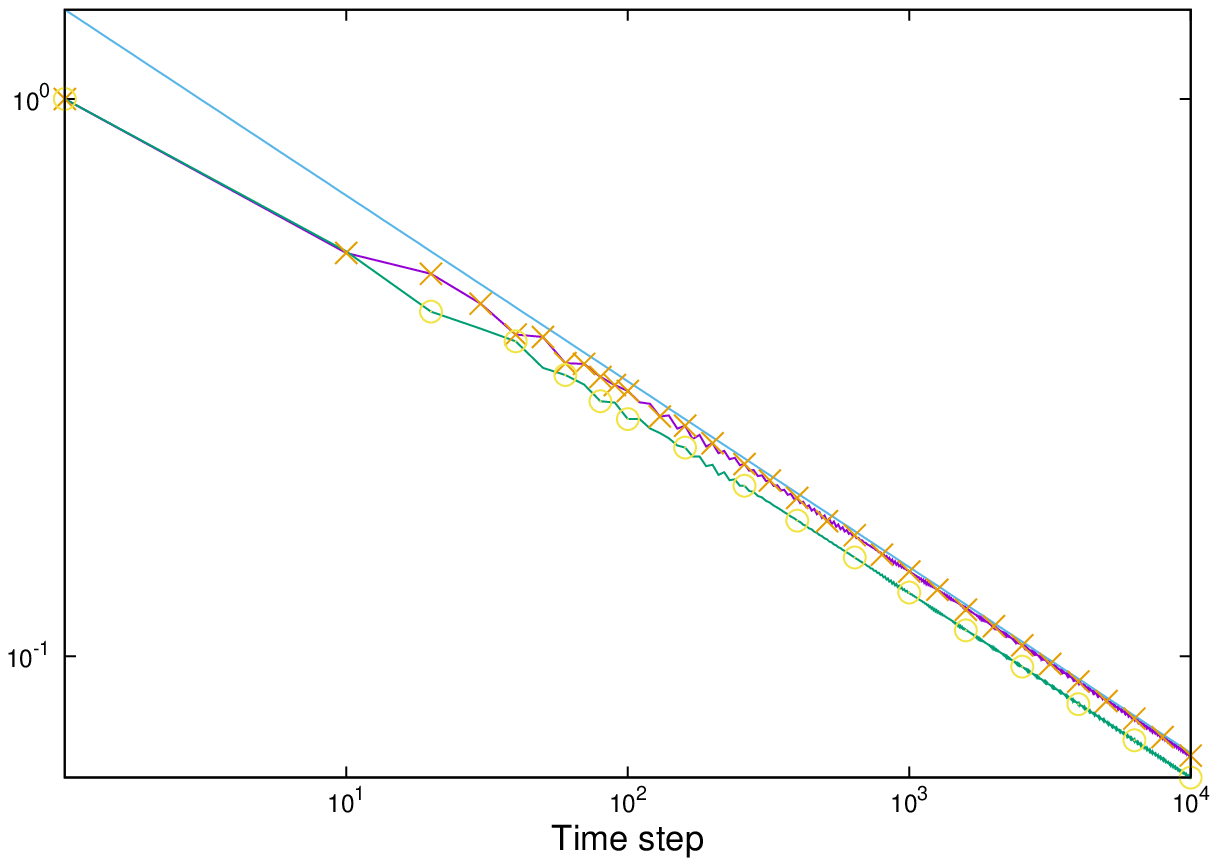}
\end{center}
\end{minipage}
\hspace{6mm}
\begin{minipage}{0.45\hsize}
(b)
\vspace{-0.5cm}
\begin{center}
\includegraphics[width=6.0cm]{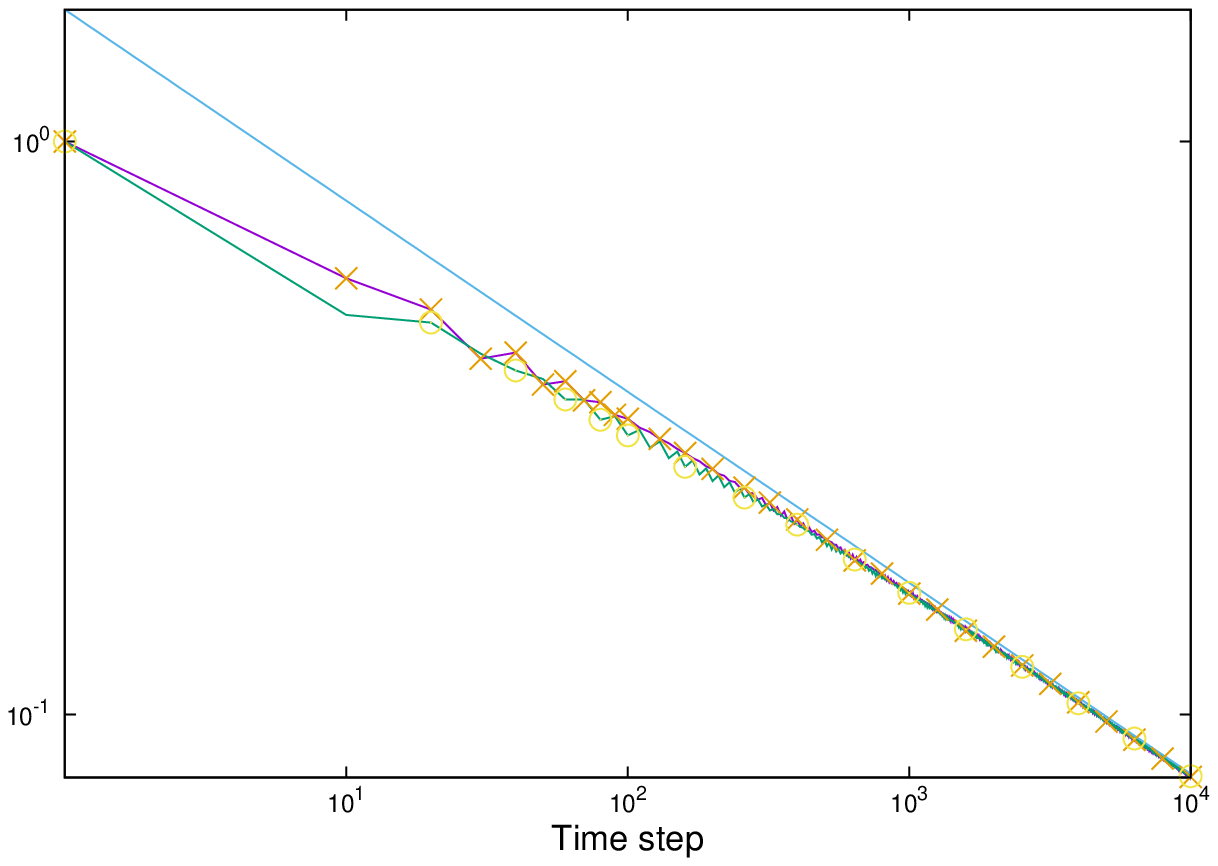}
\end{center}
\end{minipage}
\caption{Log-log plot of the behavior of $\|u\|_{l^\infty}$ to \eqref{s-eq1} with 
$p = 1$ in (a) and $p=2$ in (b). Two cases $g = 0.2\, (\times)$ and $0.4\, (\bigcirc)$ are plotted in each figure. The straight lines are given $-t/3.0 + 0.16$ in (a) and $-t/3.0 + 0.23$ in (b), which are good approximations. }
\label{log-log}
\end{figure}

We present some numerical results for nonlinear coin defined \eqref{4.5.1} in Example 1.3 to understand the dynamics of NLQW. Let $C: \mathbb{R}^2 \to U(2)$, and the following model is considered:
\begin{align}
C (s_1, s_2) = R\left(\dfrac{\pi}{4}\right)R\left(g(s_1 +  s_2)^p\right),\quad u (0, x)=\delta_{1,0} .\label{s-eq1}
\end{align}
Here, $\delta_{1, 0}$ is defined right below \eqref{A} and in particular it is $e_1$ at $x=0$ and $0$ at $x\neq 0$.

\begin{remark}
The coin in \eqref{s-eq1} is similar to the coin \eqref{4.5.1}.
Indeed, $\theta_0=\pi/4$ and $g$ of \eqref{s-eq1} corresponds to $\lambda g^p$ for $g$ in \eqref{4.5.1}.
Thus, $g$ in \eqref{s-eq1} can be positive or negative.
\end{remark}

By Theorem \ref{thm:scat}, we know that for $p \ge 2$ and $|g|\ll1 $, $u(t)$ scatters.
On the other hand, we can find explicit nonscattering solutions such as 
\begin{align}\label{soliton}
\varphi(t)=a \delta_{1,-t},\quad \frac{\pi}{4}+g a^{2p}=0,
\end{align}
for $g<0$.
\begin{remark}
Actually, it suffices to have $\frac{\pi}{4}+g a^{2p}\in 2\pi n$ so even for $g>0$, there exists the above kind of solution.
Further, if $\frac{\pi}{4}+g a^{2p}=\pi/2$, there is a periodic solution with period $4$, which evolves such as
$a \delta_{1,0}\to a \delta_{2,1}\to -a \delta_{1,0}\to -a \delta_{2,1} \to a \delta_{1,0}\to \cdots$.
\end{remark}

By elementary argument we can show $\varphi$ is unstable.
Indeed, if we take initial data $\varphi_{\varepsilon}(0):=a(1-\varepsilon)\delta_{1,0}$, then as $\varepsilon\downarrow 0$, $\varphi_\varepsilon\to \varphi(0)$.
However, setting $\varphi_\varepsilon(t)$ be the solution of NLQW, since 
\begin{align*}
\|\varphi_{\varepsilon}(t,-t)\|_{\C^2}&=|\cos\(\frac{\pi}{4}+g \|\varphi_\varepsilon(t-1,-(t-1))\|_{\C^2}^{2p}\)| \|\varphi_\varepsilon (t-1,-(t-1))\|_{\C^2}\\&\leq (1-\varepsilon')\|\varphi_\varepsilon (t-1,-(t-1))\|_{\C^2},
\end{align*}
for some $\varepsilon'>0$ depending only on $\varepsilon$, 
we see that $\|\varphi_\varepsilon(t,-t)\|_{\C^2}\to 0$ and obtain $\|\varphi(t)-\varphi_\varepsilon(t)\|_{l^2}\to \|\varphi(t)\|_{l^2}+\|\varphi_\varepsilon(t)\|_{l^2}\sim 1$.
Here, notice that $\varphi_{\varepsilon}(t,-t)$ always has the form $\begin{pmatrix} * \\ 0 \end{pmatrix}$.

\begin{remark}
We do not exclude the possibility that $\varphi_\varepsilon$ converges (in some sense) to another unknown traveling wave type solution.
\end{remark}

On the other hand, if we start from the initial data
\begin{align*}
\psi_\varepsilon(0,x)=\begin{cases} 0 & x<0\\ (1+\varepsilon)\begin{pmatrix} a \\ 0 \end{pmatrix} & x=0\\
\text{whatever} & x>0,
\end{cases}
\end{align*}
then by an easy computation, we see that $\psi_\varepsilon(t,-t)\to \begin{pmatrix} a \\ 0 \end{pmatrix}$.
Thus, there is a large set of initial data which persists the speed $1$ soliton like behavior in the left edge.

In the following, we treat \eqref{s-eq1} with $p = 1,2$ and $|g|$ not necessary small to understand its dynamics of NLQW which we could not rigorously study.   
In the case $|g|$ is not small, it is expected that nonlinear effects arise. It will turn out that the behavior of $\|u\|_{l^\infty}$ is different between the case $g > 0$ and $g < 0$. 

\subsection{Decay of $l^\infty$ norm}

First, we show the behavior of $\|u\|_{l^\infty}$ in Figure \ref{sup-norm-u}. 
As examples, we use $|g| = 0, 0.2, 0.4, 0.6, 0.8, 1$ and show the behaviors of $\|u\|_{l^\infty}$ of \eqref{s-eq1} with $p = 1$ (Figure \ref{sup-norm-u} (a) and (c)) and $p = 2$ (Figure \ref{sup-norm-u} (b) and (d)). From Figure \ref{sup-norm-u}, it is easily seen that $\|u\|_{l^\infty}$ decays like a linear for small $|g|$. From Theorem \ref{thm:1}, the behavior of $\|u\|_{l^\infty}$ of the linear model is approximated by $t^{-1/3}$ as $t \to +\infty$. We verify that the behavior of $\|u\|_{l^\infty}$ to the nonlinear model \eqref{s-eq1} is linear if the nonlinear effect $|g|$ is sufficiently small.  In Figure \ref{log-log}, it is demonstrated that the cases $g = 0.2$ and $g = 0.4$ are approximated by a linear model.

On the other hand, if $|g|$ is close to $1$, then $\|u\|_\infty$ becomes almost constant. Indeed, for the cases $g = 0.8$ and $g = 1$ in Figure \ref{sup-norm-u} (a) and for the case $g = 0.8$ in Figure \ref{sup-norm-u} (b), we see that $\|u\|_{l^\infty}$ is almost constant. Similar behaviors can be seen in Figure \ref{sup-norm-u} (c) and (d). A striking feature of the case $g < 0$ is that an oscillating behavior of $\|u\|_{l^\infty}$ appears. We can see oscillations of $\|u\|_{l^\infty}$ if $g = -0.6, -0.4$ in Figure \ref{sup-norm-u} (c) and if $g = -0.6$ in Figure \ref{sup-norm-u} (d). We expect that this oscillation starts decaying lineally as time step $t$ becomes larger. We would like to discuss more details of such cases later.

\subsection{Soliton and oscillating behavior of $u$}\label{sec:non-decay}

\begin{figure}[t]
\begin{tabular}{cc}
\begin{minipage}{0.45\hsize}
(a)
\vspace{-0.5cm}
\begin{center}
\includegraphics[width=6.0cm]{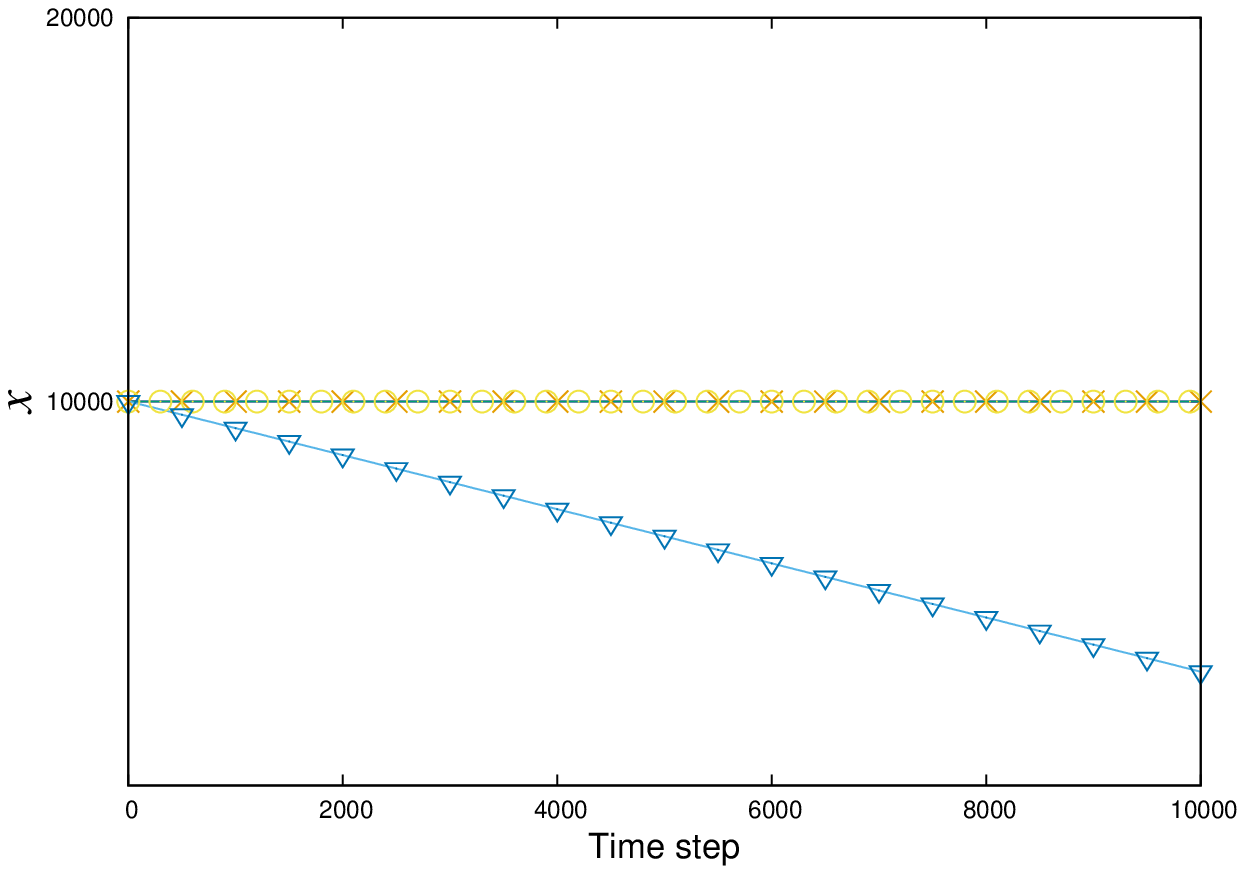} 
\end{center}
\end{minipage}
\hspace{6mm}
\begin{minipage}{0.45\hsize}
(b)
\vspace{-0.5cm}
\begin{center}
\includegraphics[width=6.0cm]{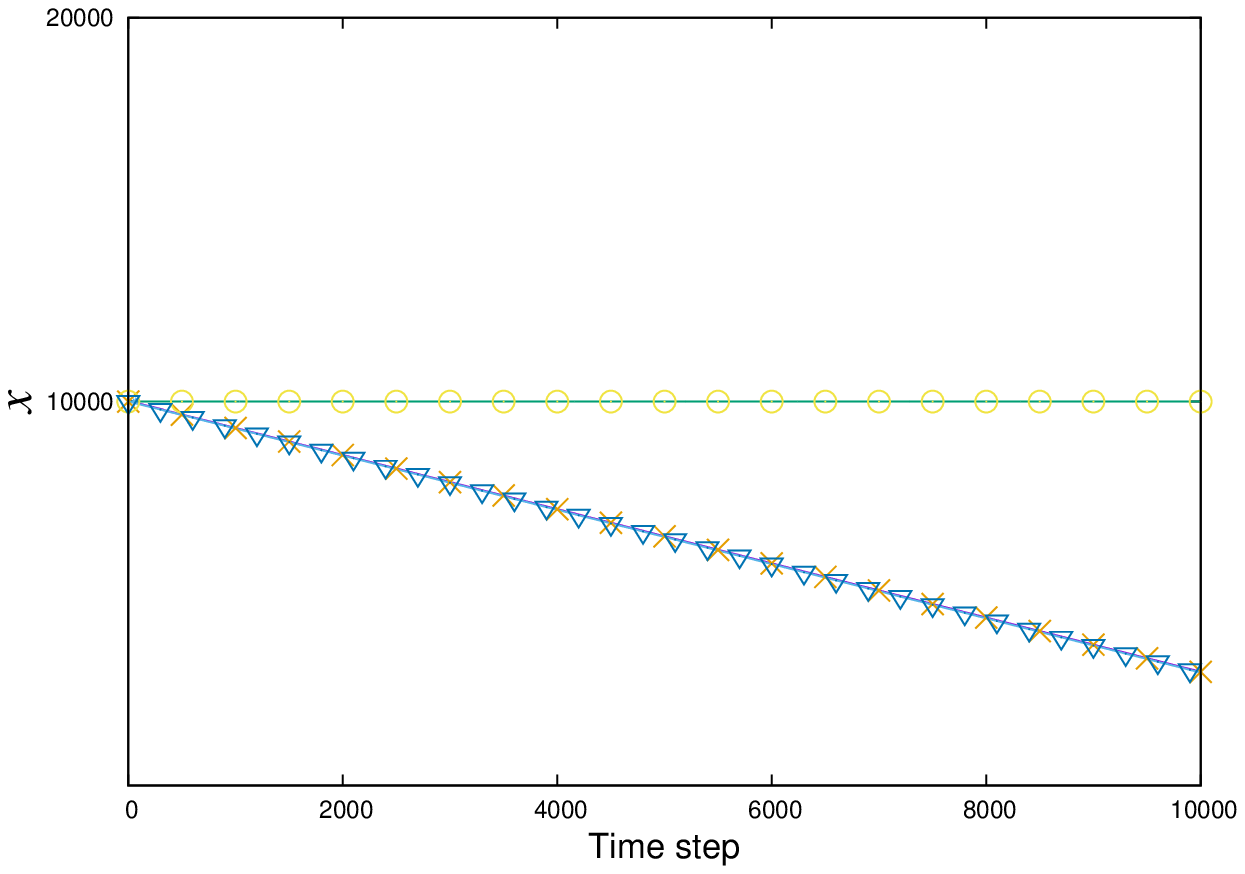}
\end{center}
\end{minipage}\vspace{0.3cm}\\
\begin{minipage}{0.45\hsize}
(c)
\vspace{-0.5cm}
\begin{center}
\includegraphics[width=6.0cm]{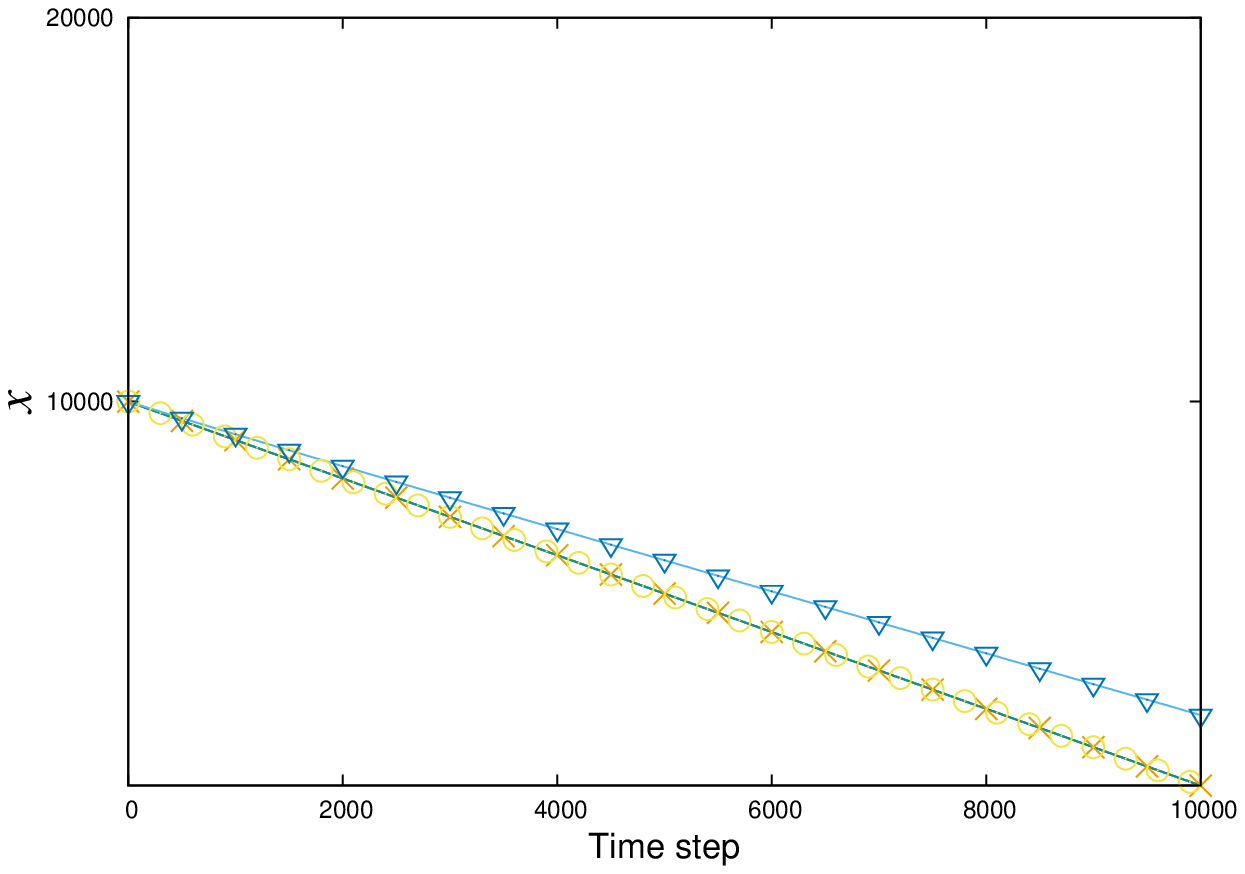}
\end{center}
\end{minipage}
\hspace{6mm}
\begin{minipage}{0.45\hsize}
(d)
\vspace{-0.5cm}
\begin{center}
\includegraphics[width=6.0cm]{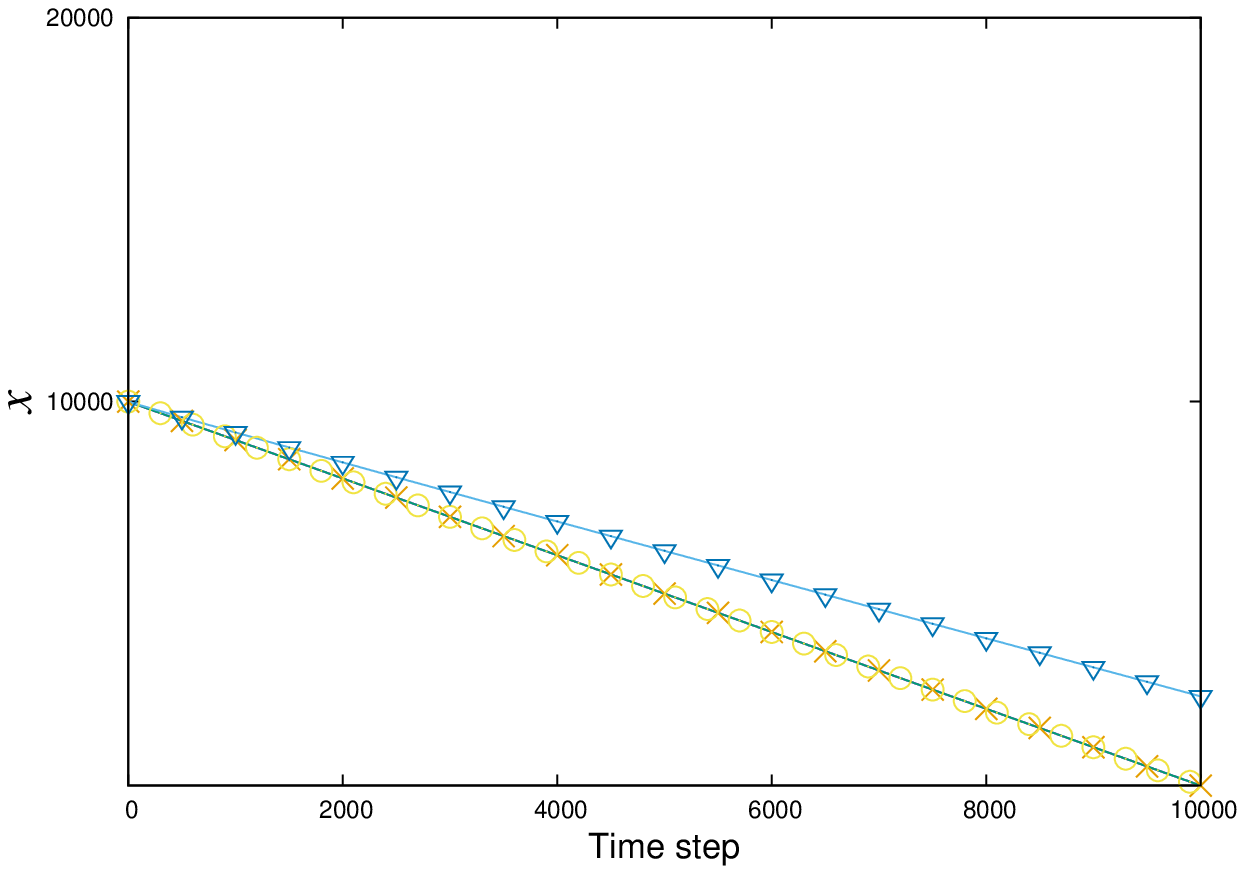}
\end{center}
\end{minipage}
\end{tabular}
\caption{Trace of the maximum point of $u$, where the vertical axis represents the place $x \, (0 \le x \le 20000)$ and and the horizontal axis is time step. At $t = 0$, $u(0, x) = \delta_{1, 10000}$. Figures (a) and (b) are the case $g > 0$ for $p = 1$ and $p = 2$, respectively. Figures (c) and (d) are the case $g < 0$ for $p = 1$ and $p = 2$, 
respectively. In these figures, we put $g = \pm1\ (\times)$, $g = \pm 0.8\ (\bigcirc)$, $g = \pm 0.6 \ (\triangledown)$. }
\label{location-u}
\end{figure}

Figure \ref{sup-norm-u} shows that $\|u\|_{l^\infty}$ does not decay lineally if $|g|$ is not sufficiently small. When $\|u\|_{l^\infty}$ becomes almost constant, we can obtain the value of $a = \|u\|_{l^\infty}$ by solving the equations $\pi/4 + g a^{2p}=\pi/2$ for $g>0$ and $\pi/4+g a^{2p}=0$ for $g<0$. In Table 1, we compare the value of $\left(\pi/4|g|\right)^{1/2p}$ and the value of $\|u\|_{l^\infty}$ obtained by the numerical calculations.  
It is noticed that the numerical result of $\|u\|_{l^\infty} $is much smaller than $(\pi/4|g|)^{1/2p}$ in the case $p = 2$ and $g = 1.0$. From Figure \ref{sup-norm-u} (b), $\|u\|_{l^\infty}$ decays like a linear model in this case, which is the curve with the mark `$\times$'.


\begin{table}[h!]
\begin{center}
\begin{tabular}{|c|c|c|c|}
\hline
 &$g$ & $(\pi/4|g|)^{1/2p}$ & Numerical results of $\|u\|_{l^\infty}$\\
\hline
\multirow{4}{*}{$p=1$} & $ 0.8$ &\multirow{2}{*}{$0.9908318244$} & 0.990865\\
\cline{2-2} \cline{4-4}
& $ -0.8$ & & 0.990911\\
\cline{2-4}
 & $ 1.0$ & \multirow{2}{*}{$0.88622692545$} & 0.886256 \\
\cline{2-2} \cline{4-4}
 & $ -1.0$ & & 0.886299\\
\hline
\multirow{4}{*}{$p=2$}  & $ 0.8$ & \multirow{2}{*}{$0.99540535682$} & 0.995414\\
\cline{2-2} \cline{4-4}
 & $-0.8$ & & 0.995425\\
\cline{2-4}
 & $ 1.0$ & \multirow{2}{*}{$0.94139626377$} & 0.111958\\
\cline{2-2} \cline{4-4}
 & $ -1.0$ & & 0.941415\\
\hline
\end{tabular}
\caption{Numerical results are obtained at $t = 10000$. }
\end{center}
\end{table}

In Figure \ref{location-u}, a location of the maximum point of $u$ is traced, which starts from $u(0, x) = \delta_{1, 10000}$. 
From Figure \ref{location-u} (a) and (b), we see that the trace of the maximum point of $u$  is almost constant for $g = 1$ and $g= 0.8$ in (a) and for $g = 1$ in (b).   
Figure \ref{location-u} (c) and (d) show that a location of the maximum point of $u$ moves to the left on the $x$-axis at a constant speed in every case. 

For the  case $g < 0$, we have already seen in Figure \ref{sup-norm-u} (c) and (d) that $\|u\|_{l^\infty}$ is almost constant if $g = -0.8$ and it oscillates if $g = -0.6$. Figure \ref{location-u} (c) and (d) say that the maximum point of $u$ moves similarly except speed in both cases $g = -0.8$ and $g = -0.6$. We would like to understand the dynamics of these cases. In Figure \ref{soliton}, we present the behavior of $u$ in the case $p = 2$ and $g = -0.6$ and the case $p = 2$ and $g = -0.8$. When $p = 2$ and $g = -0.8$ (the second line of Figure \ref{soliton}), the soliton moves to the left. However, $u$ moves up and down when $p = 2$ and $g= -0.6$ (the first line of Figure \ref{soliton}).



\begin{figure}[t]
\begin{tabular}{cccc}
\includegraphics[width=3.5cm]{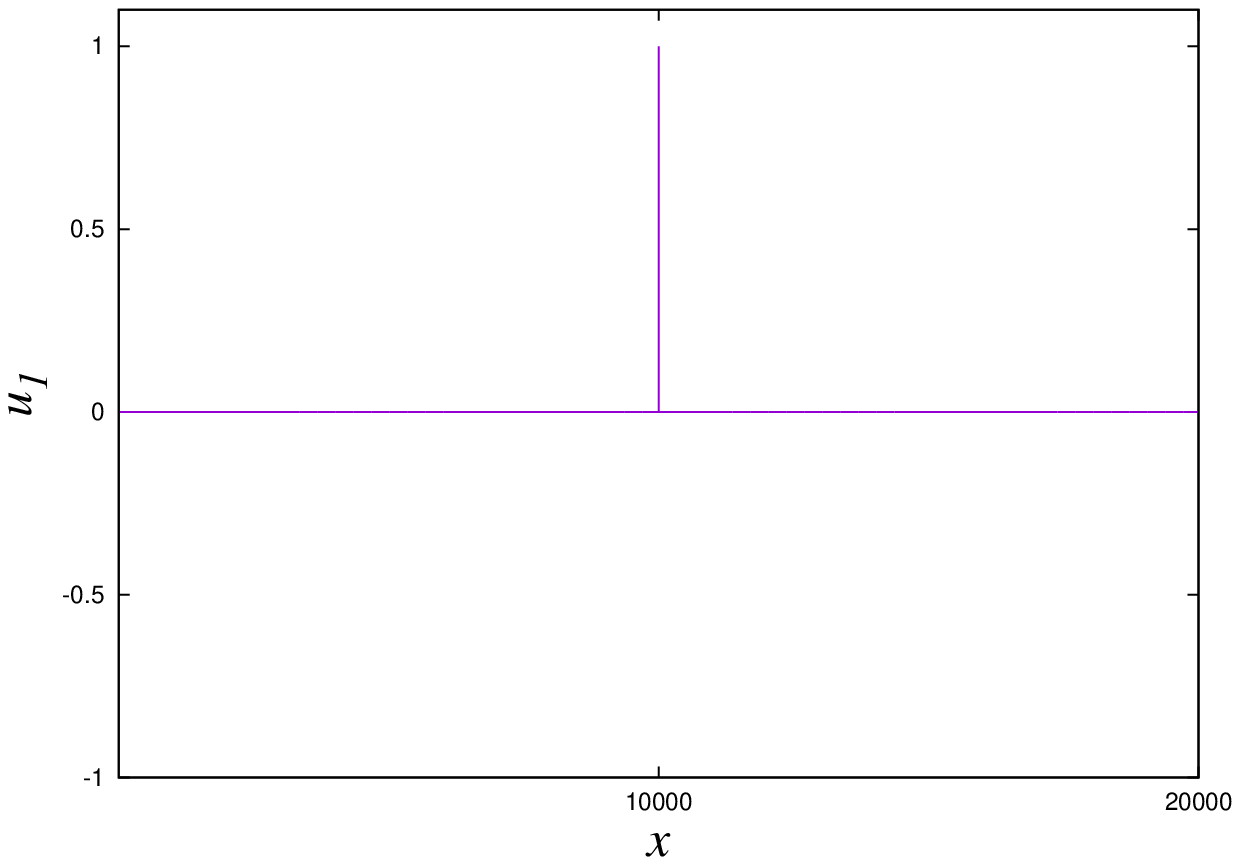} &

\includegraphics[width=3.5cm]{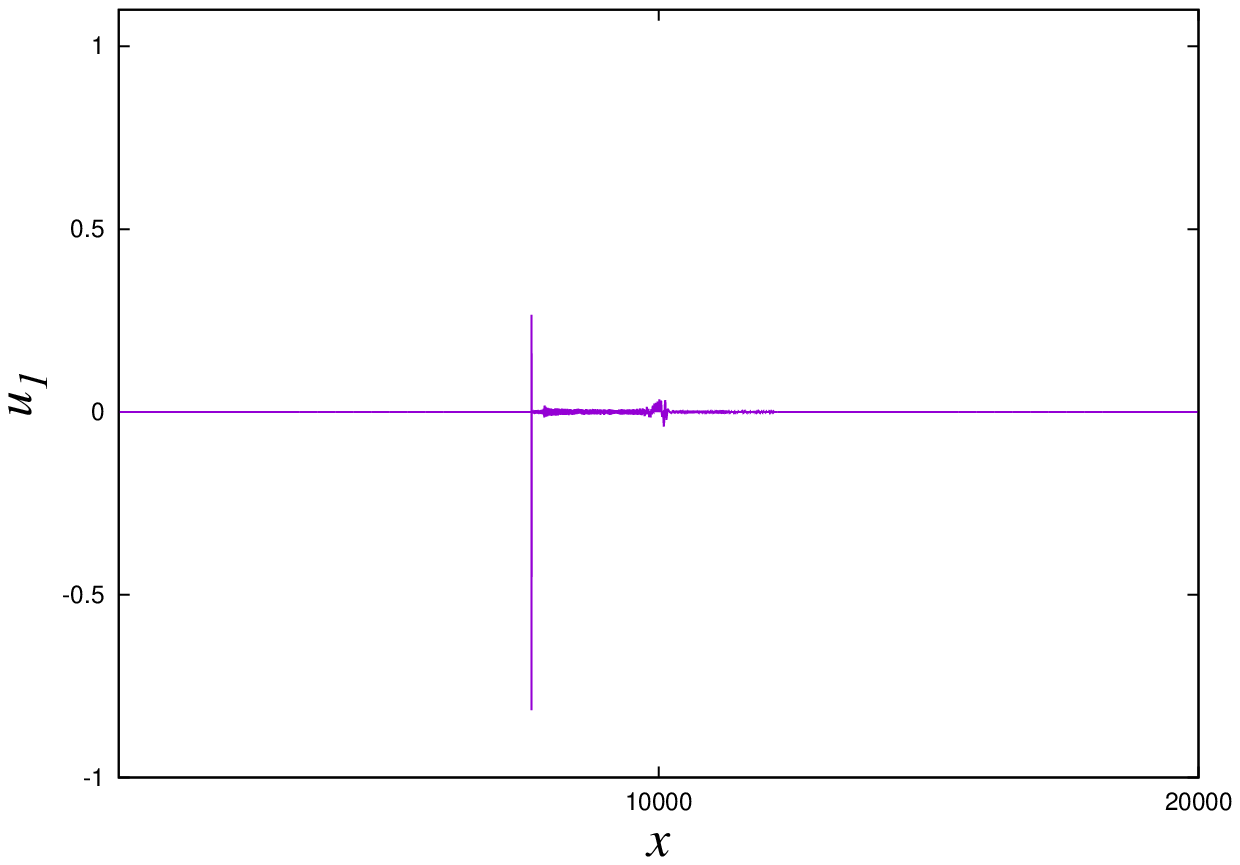} &

\includegraphics[width=3.5cm]{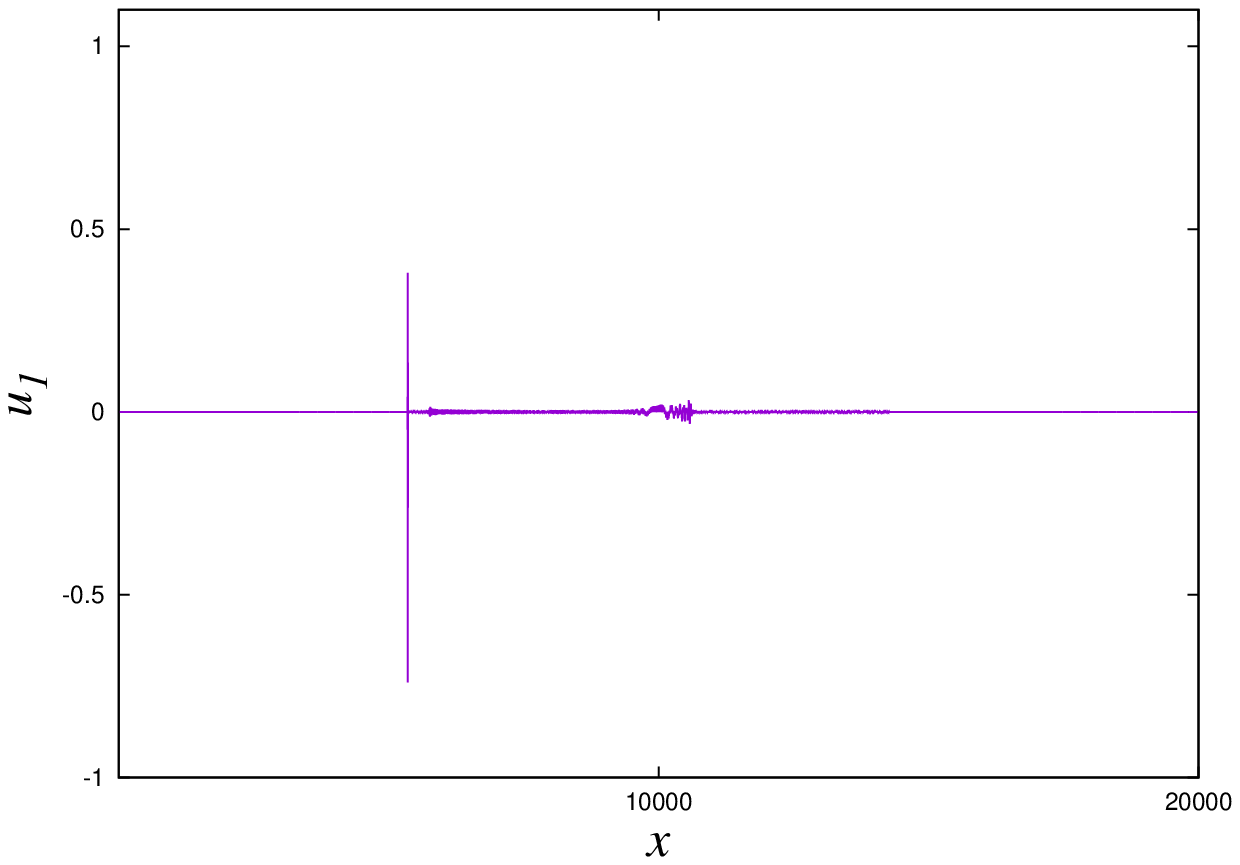}&

\includegraphics[width=3.5cm]{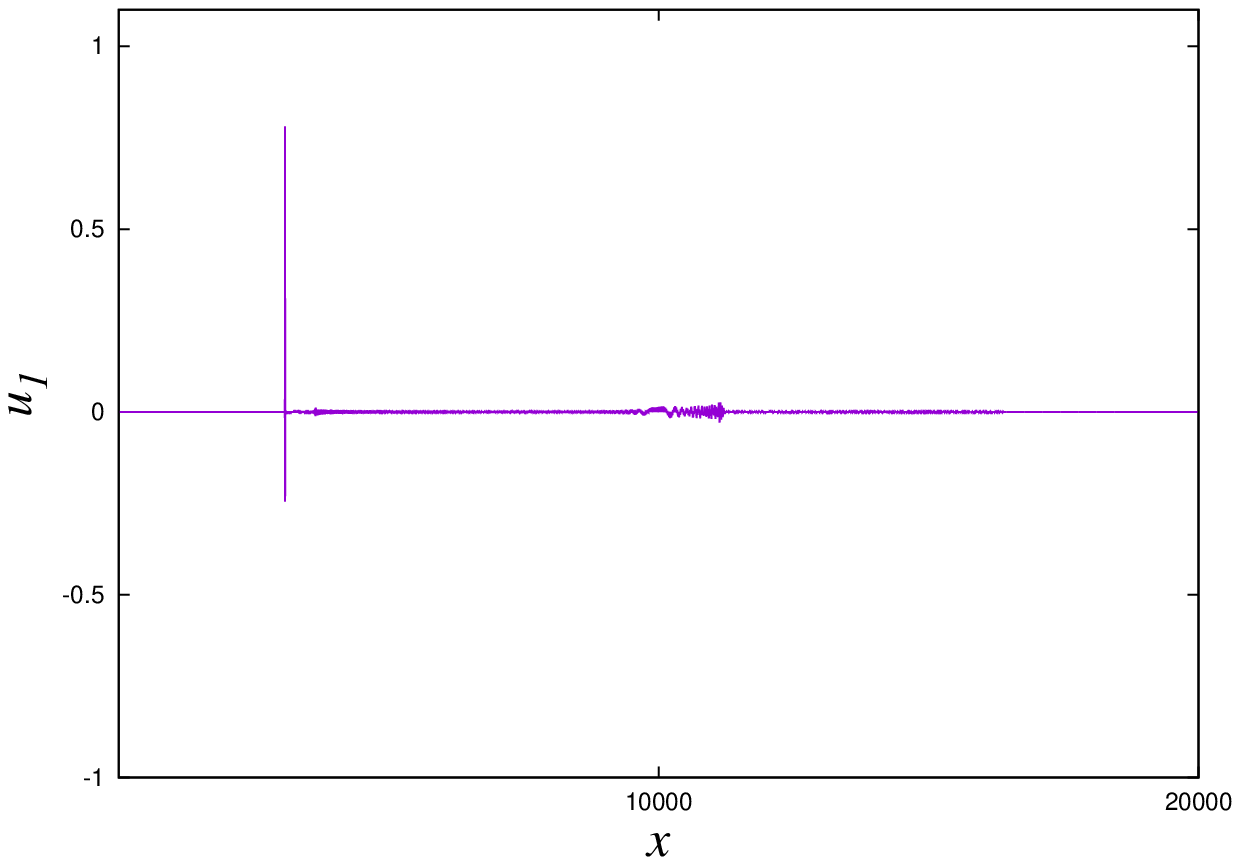} \\

\includegraphics[width=3.5cm]{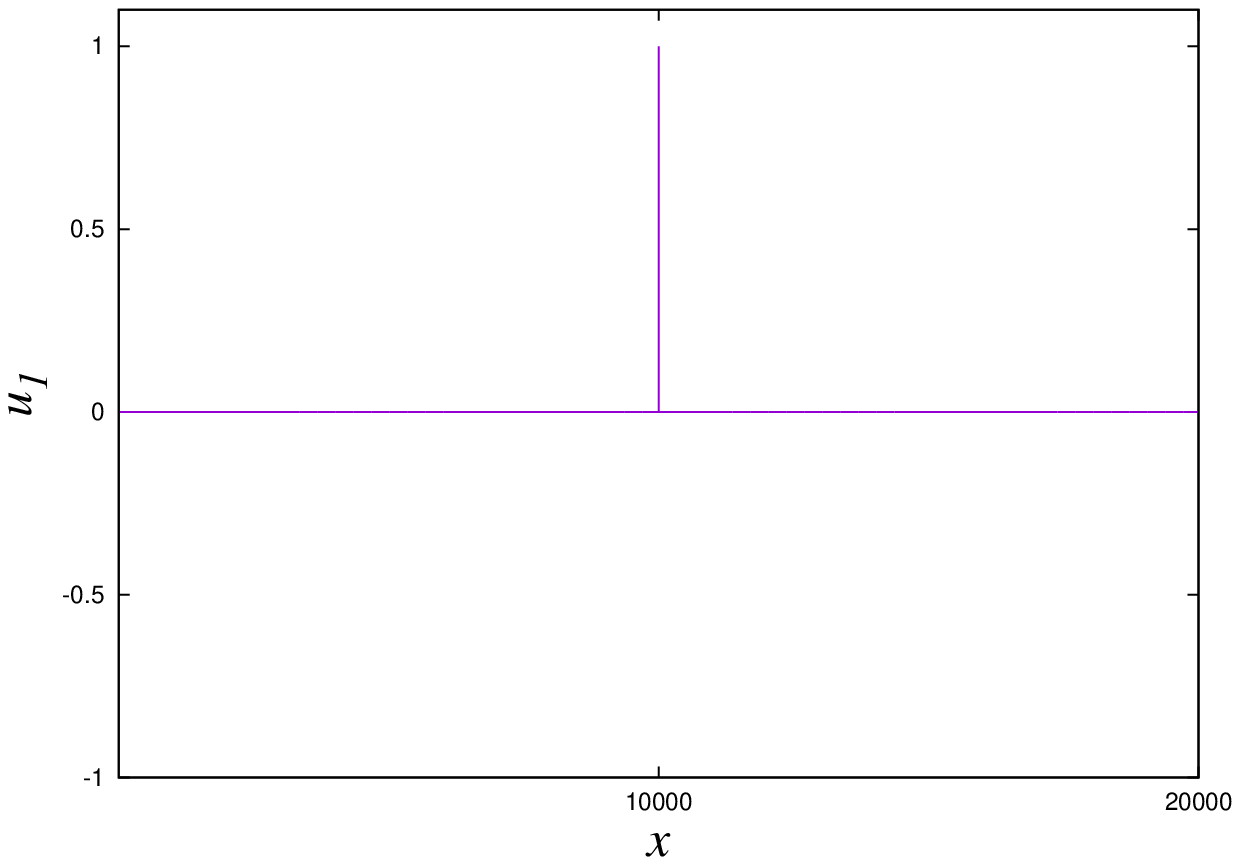} &

\includegraphics[width=3.5cm]{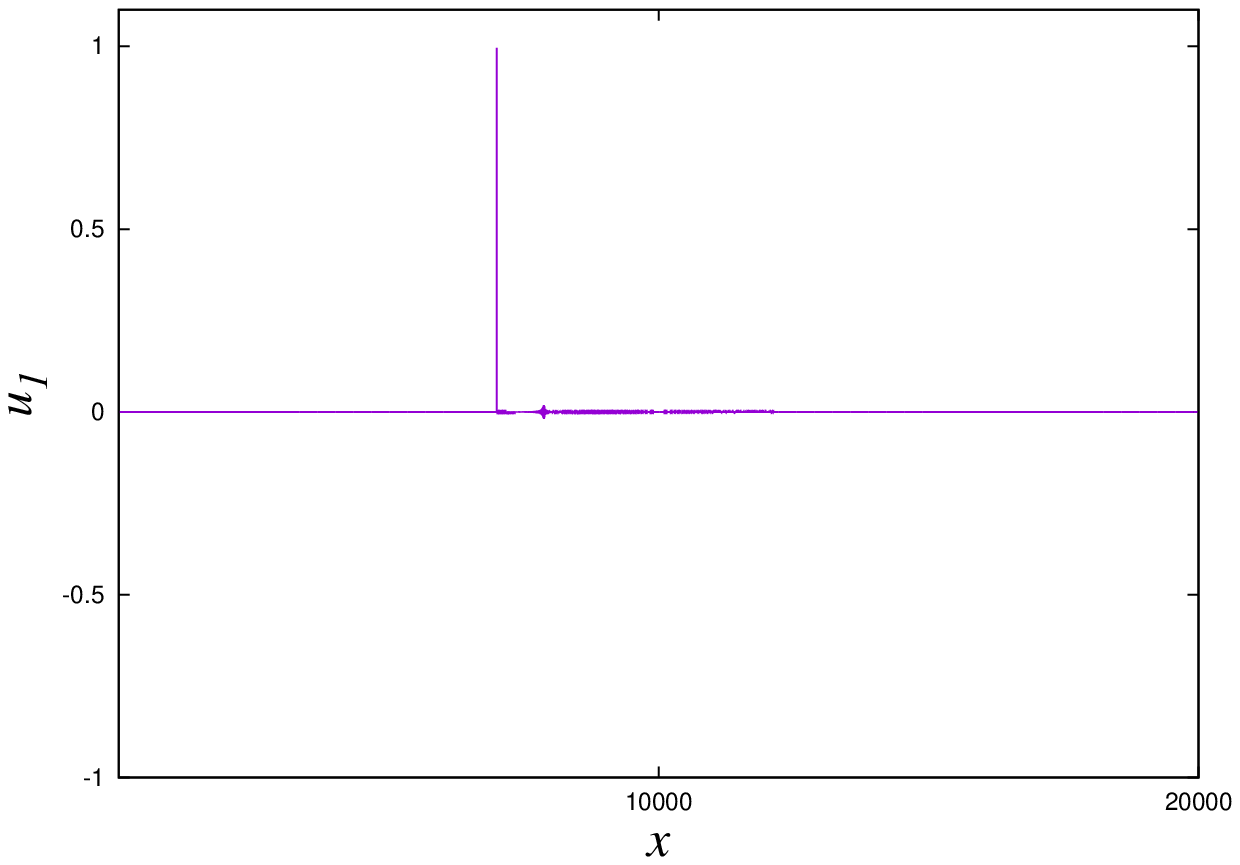} &

\includegraphics[width=3.5cm]{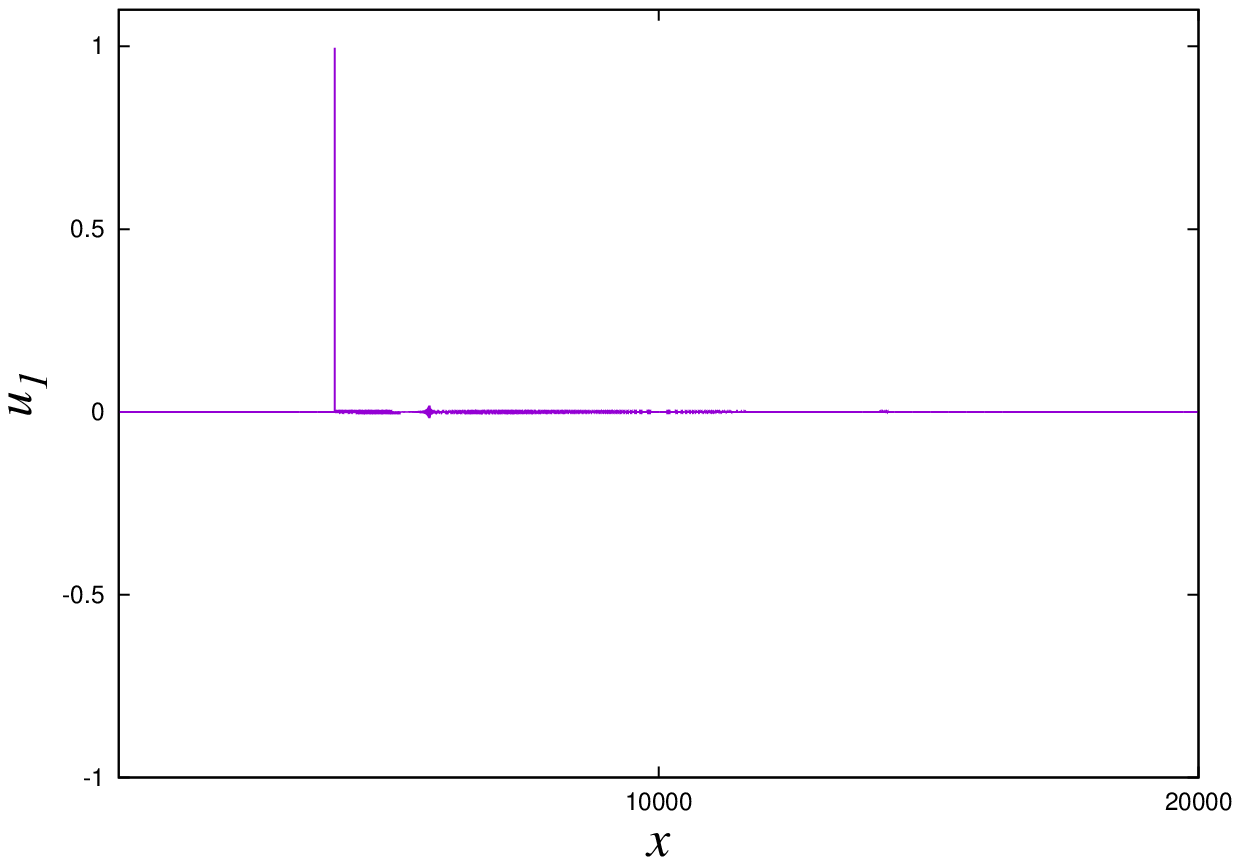}&

\includegraphics[width=3.5cm]{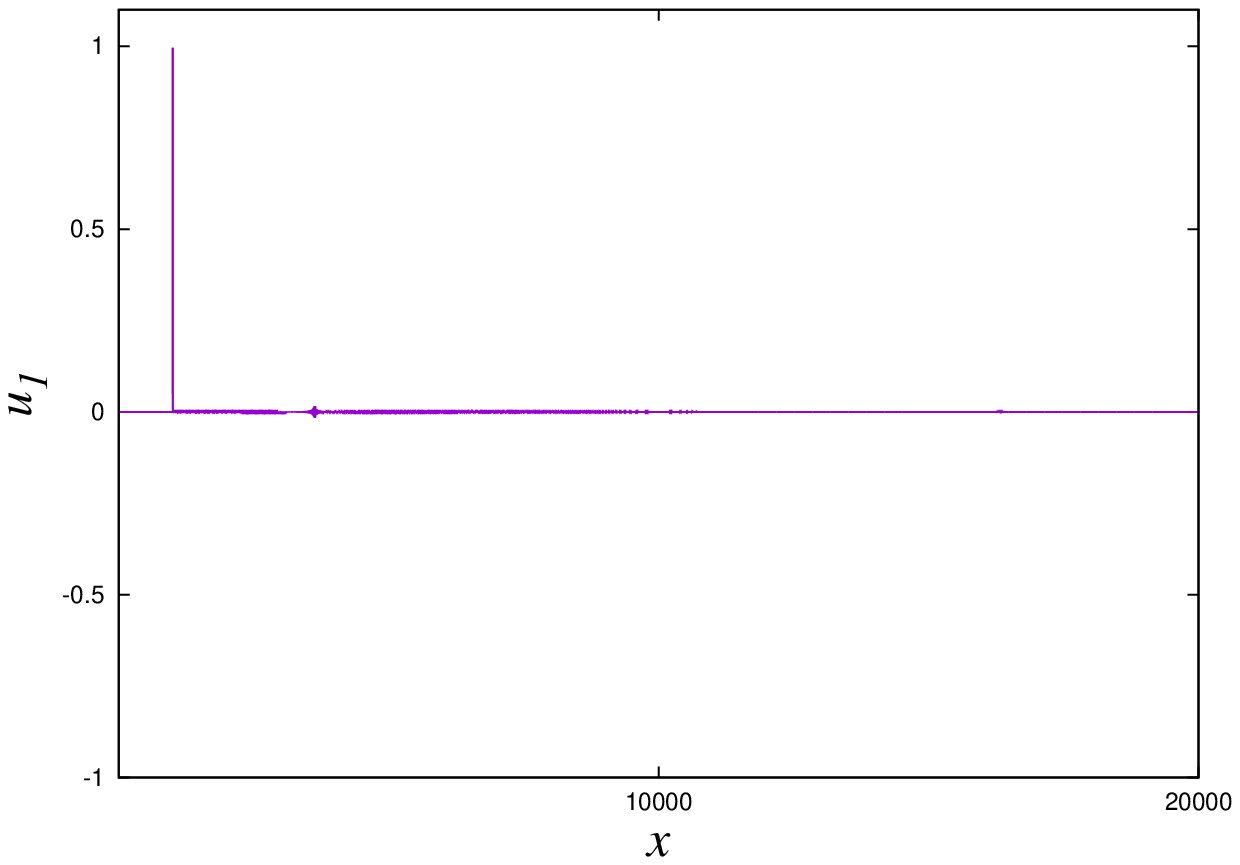}

\end{tabular}
\caption{The first line, four figures are in the case $p = 2$ and $g = -0.6$. The second line, they are in the case $p = 2$ and $g = -0.8$. The initial condition is $u(0, x) = \delta_{1, 10000}$. In both cases, they are at $t = 0, 3000, 6000, 9000$ from the left. We plot only the first element of $u = (u_1, u_2)^t$. }
\label{soliton}
\end{figure}


\subsection{Strong nonlinear regime}

Finally, we give an example to see what happens if $|g|$ is sufficiently large. When $|g|$ is large, it is expected that the behavior of $\|u\|_{l^\infty}$ becomes complicated. We treat the model \eqref{s-eq1} with $p = 2$ and consider the case $g = -15.2$. Figure \ref{max-7} shows the behavior of $\|u\|_{l^\infty}$ with the initial condition $u(0, x) = \delta_{1, 10000}$. Since $\|u(0, \cdot)\|_{l^\infty} = 1$, we see that it decreases quickly at the beginning and stays around $0.4$. Since it has small oscillations, we would like to compare this case with the case $p = 2$ and $g = -0.6$ where we have an oscillation of $\|u\|_{l^\infty}$. 
We trace points on which  the first element of $u = (u_1, u_2)^t$ is larger than $0.1$ in Figure \ref{max-8}. There are three lines in Figure \ref{max-8}. This implies that there exist at least three solitons or oscillating behaviors of $u$. We show the behavior of $u$ in Figure \ref{g-big}, where three oscillating peaks appear. 


\begin{figure}[t]
\begin{minipage}{0.45\hsize}
\begin{center}
\includegraphics[width=6.0cm]{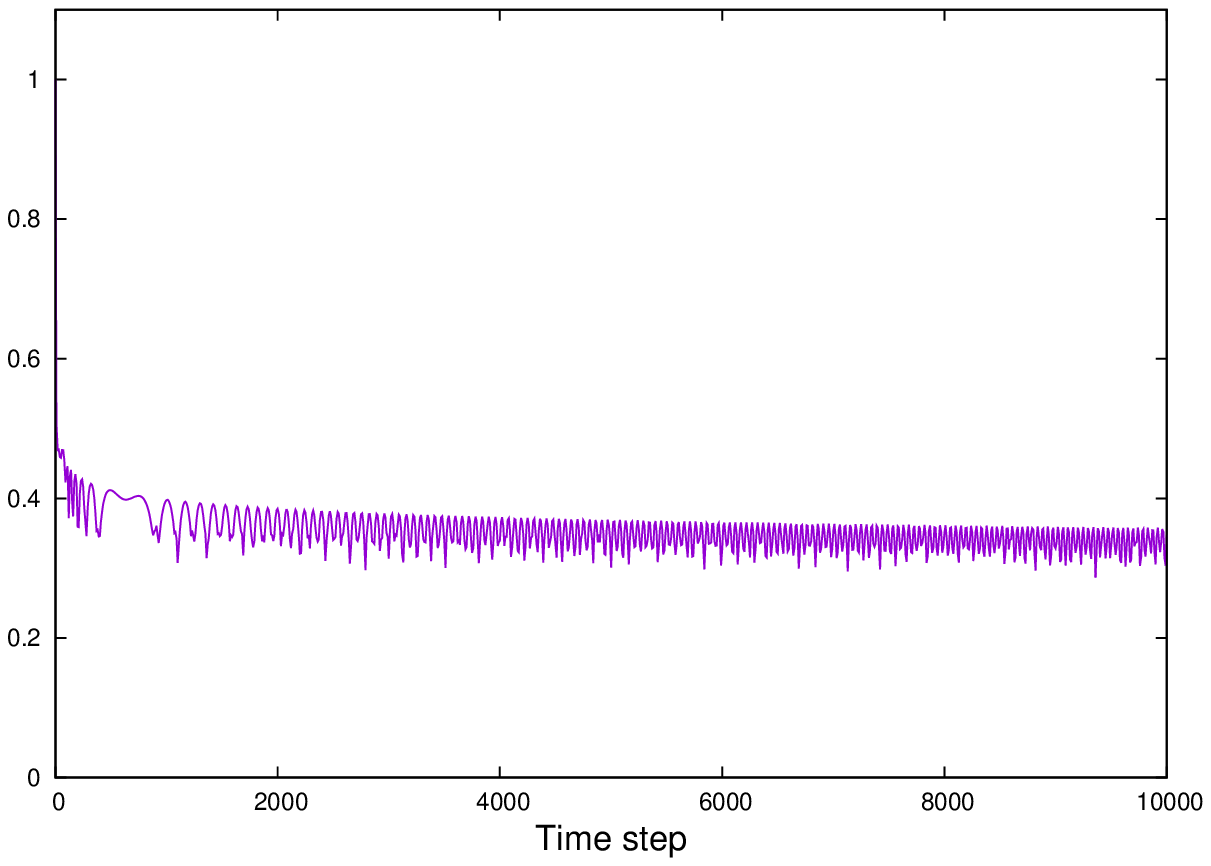}
\caption{The behavior of $\|u\|_{l^\infty}$ for $g = -15.2$ and $p = 2$.}
\label{max-7}
\end{center}
\end{minipage}
\hspace{6mm}
\begin{minipage}{0.45\hsize}
\begin{center}
\includegraphics[width=6.0cm]{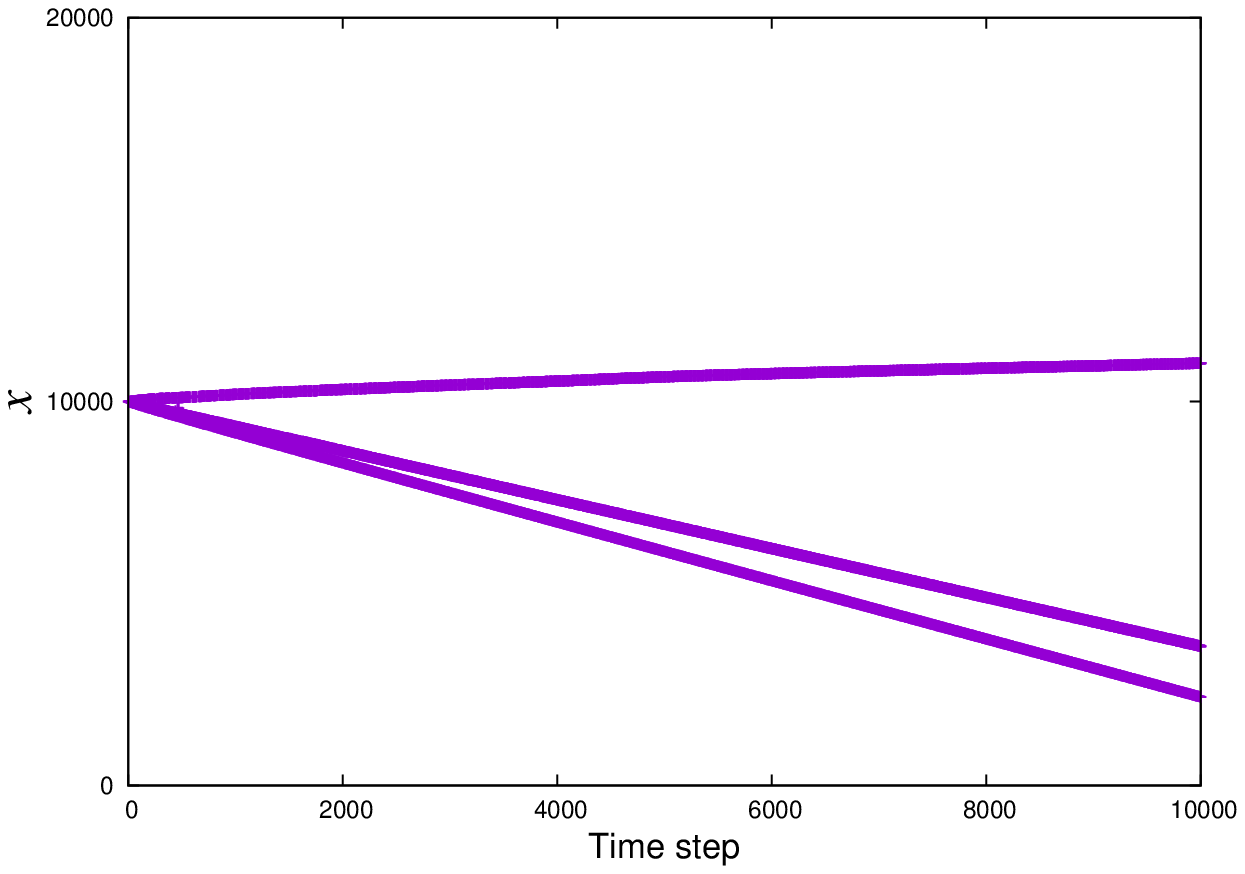}
\caption{The traces of points on which the first element of $u = {}^t\!(u_1, u_2)$ is larger than $0.1$ for $g = -15.2$ and $p = 2$. }
\label{max-8}
\end{center}
\end{minipage}
\end{figure}


\begin{figure}[h!]
\begin{tabular}{cccc}
\includegraphics[width=3.5cm]{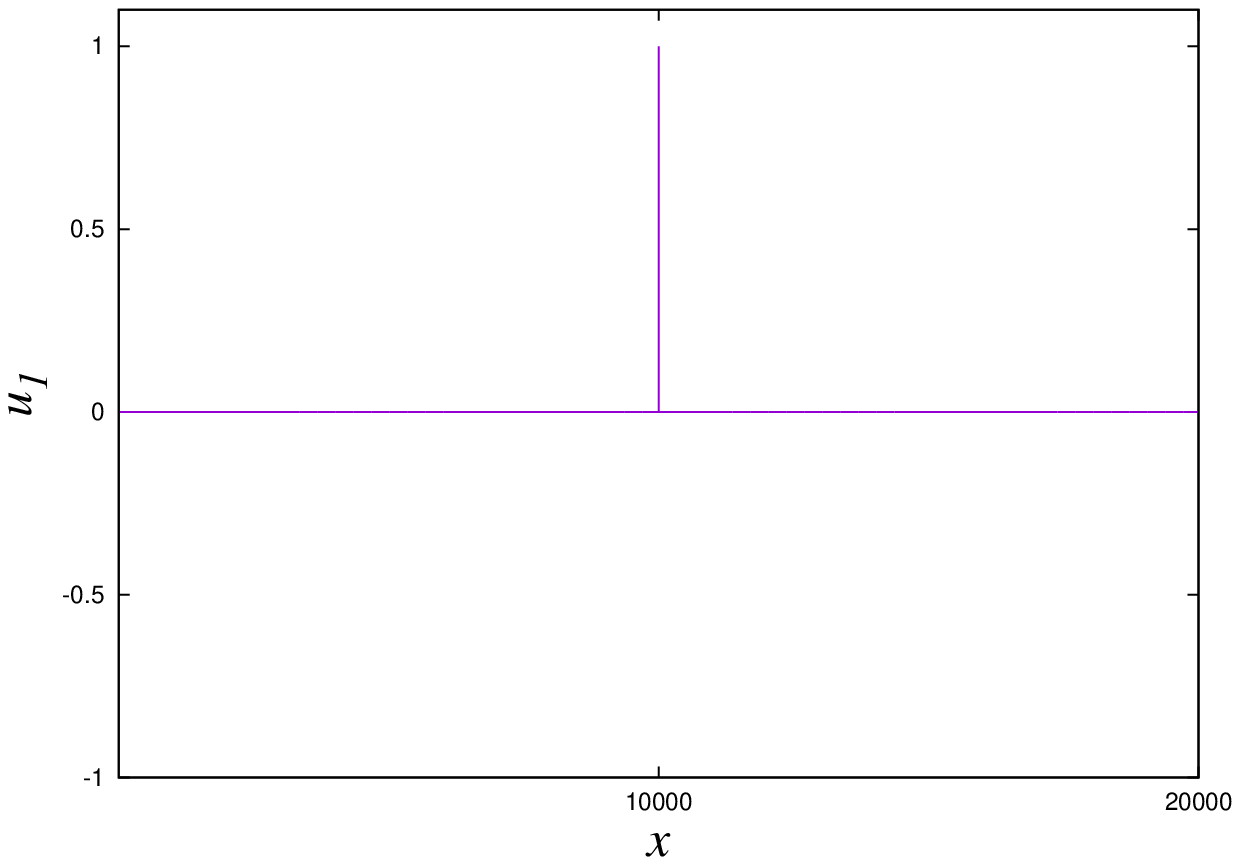} &

\includegraphics[width=3.5cm]{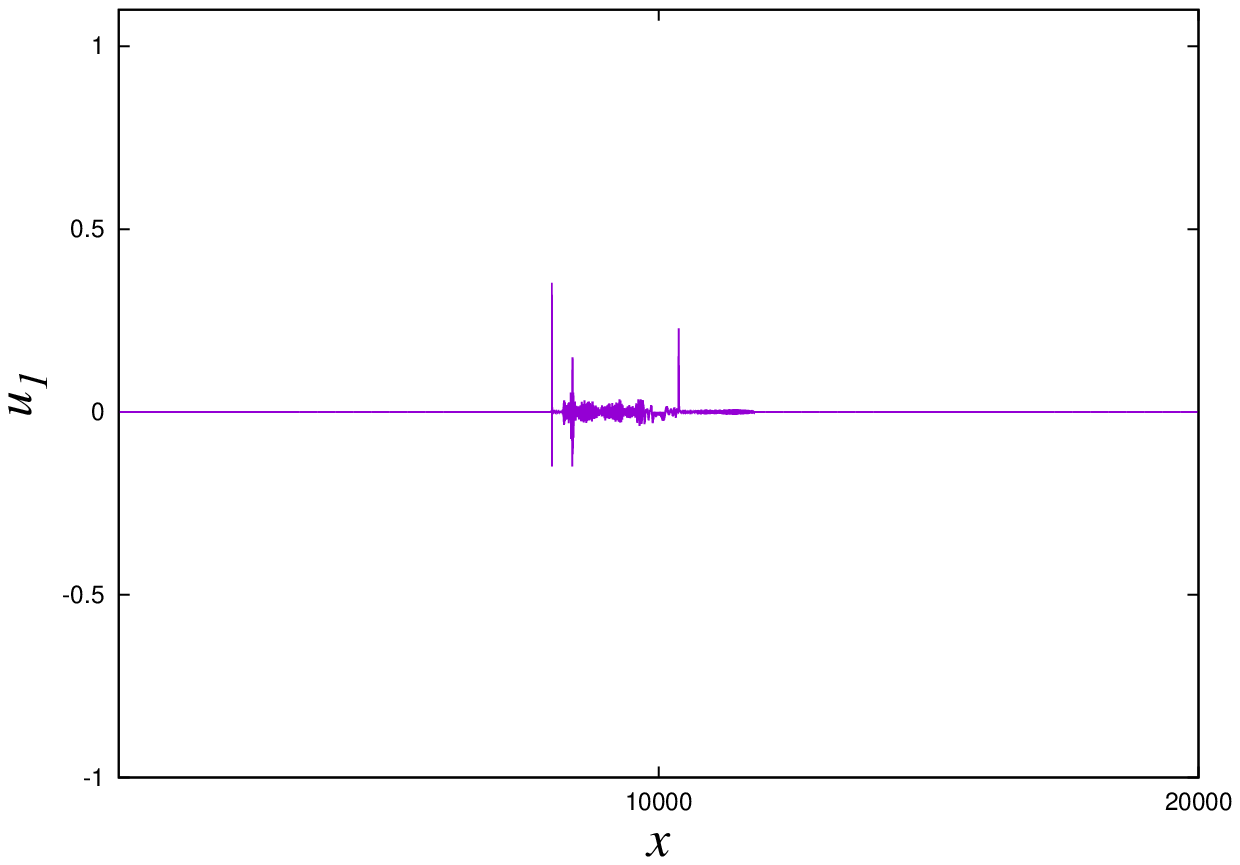} &

\includegraphics[width=3.5cm]{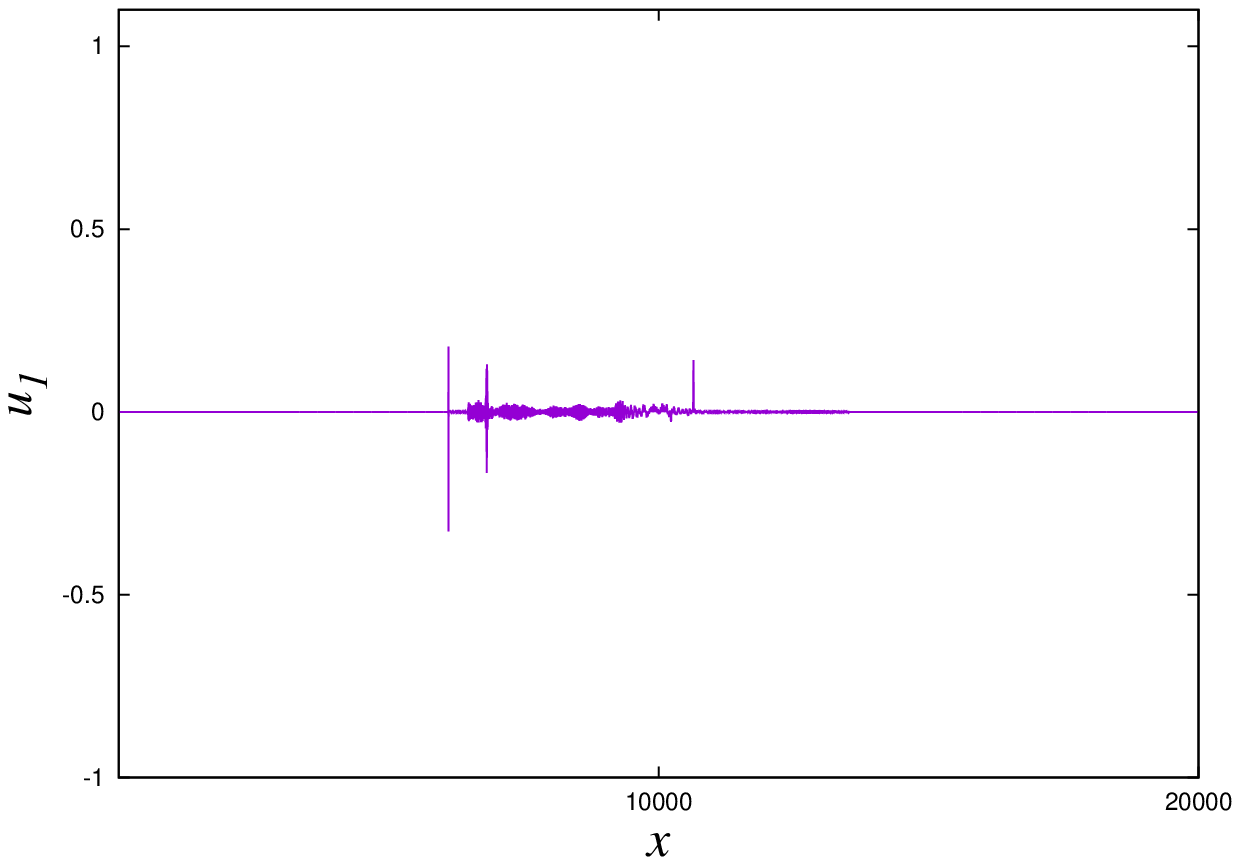}&

\includegraphics[width=3.5cm]{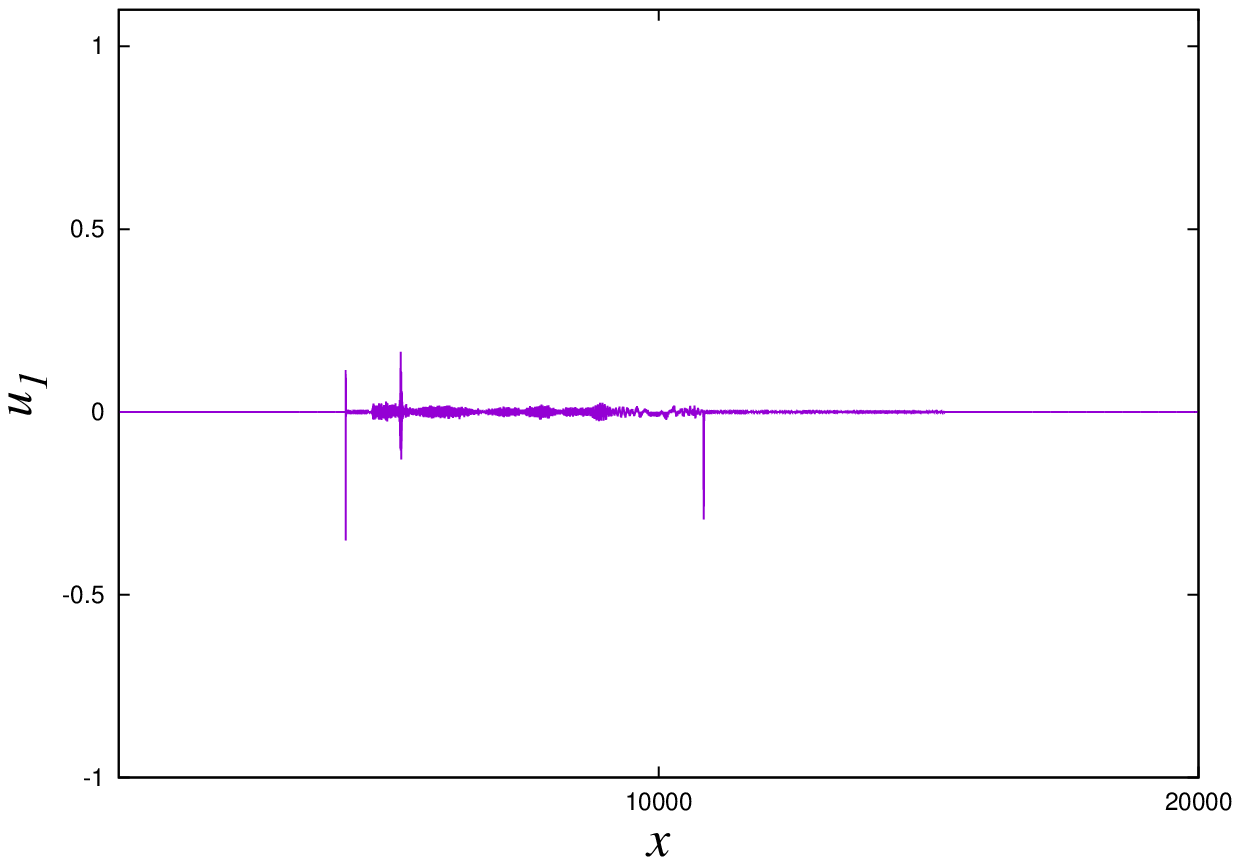}

\end{tabular}
\caption{$g = -15.2$ and $p = 2$. We show the behavior of the first element of $u = (u_1, u_2)^t$ corresponding to Figure \ref{max-8}. They are at $t = 0, 2500, 5000, 7500$ from the left.}
\label{g-big}
\end{figure}

\appendix
\section{Proof of the weak limit theorem}

We first prove Proposition \ref{p_wlt}. 
We suppose that $\|u_0\|_{l^2}=1$. 
By \eqref{Xt_dist}, the characteristic function of $X_t/t$
is given by
\begin{equation}
\label{chrct_X_t} 
\mathbb{E}(e^{i \xi X_t/t})  
	= \langle U(t) u_0, e^{i \xi \hat x/t} U(t) u_0 \rangle,
		\quad \xi \in \mathbb{R}, 
\end{equation}
where $\hat x$ is the position operator. 
Let $\hat v_0$ be a self-adjoint operator defined as
\[ \hat v_0 = \mathcal{F}^{-1} P^{-1}(\xi) 
	\begin{pmatrix} - p^\prime(\xi+\theta_a) & 0 \\ 
		0 & p^\prime(\xi+\theta_a) \end{pmatrix}
		P(\xi) \mathcal{F}, \]
where $p^\prime(\xi)$ is given in \eqref{12.2}
and $\theta_a$ is the argument of $a$. 
$\hat v_0$ is called the asymptotic velocity. 
As shown in \cite{GJS04PRE, Suzuki16QIP}, 
the Heisenberg operator $\hat x_0(t) = U^{-t} \hat x_0 U^t$
converges to $\hat v_0$ in the following sense:
\begin{equation}
\label{slimx} 
\mbox{s-}\lim_{t \to \infty} e^{i \xi \hat x_0(t)}
	=e^{i\xi \hat v_0}, \quad \xi \in \mathbb{R}. 
\end{equation}
\begin{lemma}
\label{lem_A1}
\[ \lim_{t \to \infty}
	\langle U(t) u_0, e^{i \xi \hat x/t} U(t) u_0 \rangle
	= \langle u_+, e^{i \xi \hat v_0} u_+ \rangle \]
\end{lemma}
\begin{proof}
A direct calculation yields
\begin{align*}
& |\langle U(t) u_0, e^{i \xi \hat x/t} U(t) u_0 \rangle
	- \langle u_+, e^{i \xi \hat v_0} u_+ \rangle| \\
& \quad \leq |\langle U(t) u_0 - U_0^t u_0,
		 e^{i \xi \hat x/t} U(t) u_0 \rangle| 
	+ |\langle U_0^t u_0, 
		e^{i \xi \hat x/t}( U(t) u_0 - U_0^t u_0) \rangle| \\
& \qquad + |\langle U_0^t u_0, e^{i \xi \hat x/t}  U_0^t u_0\rangle
		- \langle u_+, e^{i \xi \hat v_0} u_+ \rangle  \rangle| \\
& \quad =: I_1(t) + I_2(t) + I_3(t).
\end{align*}
Because $e^{i\xi \hat x/t}$ and $U(t)$ preserves the norm
and $U(t)u_0$ scatters, 
$\lim_{t \to \infty} I_1(t)= \lim_{t \to \infty} I_2(t)=0$. 
By \eqref{slimx}, $\lim_{t \to \infty}I_3(t)=0$. 
Hence, the proof is completed. 
\end{proof}
By \eqref{chrct_X_t} and Lemma \ref{lem_A1},
\begin{equation}
\label{lim_chrct}
\lim_{t \to \infty} E(e^{i \xi X_t/t})
	= \langle u_+, e^{i \xi \hat v_0} u_+ \rangle
	= \int_{[-|a|, |a|]} e^{i \xi v} d \|E_{\hat v_0}(v) u_+\|^2,
\end{equation}
where we have used the spectral theorem in the last equation. 
The right-hand side in the above equation is equal to 
the characteristic function of a random variable $V$
following the probability distribution 
$\mu_V =  \|E_{\hat v_0}(\cdot) u_+\|^2$. 
This completes the proof of Proposition \ref{p_wlt}. 

In what follows, we prove \eqref{wldist}. 
By \eqref{lim_chrct}, the characteristic function of $V$ coincides with
\[ \langle u_+, e^{i \xi  \hat v_0} u_+ \rangle
= \frac{1}{2\pi} \int_{\mathbb{T}}
 \sum_{\pm} e^{\mp i \xi  p^\prime(\xi)} 
  \langle \hat u_+(\xi-\theta_a),  
  	Q_\pm(\xi) \hat u_+(\xi-\theta_a) \rangle
  d \xi. 
\]
Let $v =  \mp p^\prime(\xi)$. 
By direct calculation, 
$\sin \xi = \frac{\mp |b| v}{|a|\sqrt{1-v^2}}$.
When $-\pi/2 \leq \xi \leq \pi/2$, 
$\xi = \arcsin \frac{\mp |b| v}{|a|\sqrt{1-v^2}}$. 
When $\pi/2 \leq \xi \leq 3\pi/2$, 
$-\pi/2 \leq \xi - \pi \leq \pi/2$ and $\sin \xi = \sin (\xi -\pi)$
imply that
$\xi = \pi +  \arcsin \frac{\pm |b| v}{|a|\sqrt{1-v^2}}$.
Thus, 
\[ \frac{d \xi}{\pi} 
= 
	 f_K(v;|a|) d v. \]
Let
\[  P_{\pm, m}(\xi) = \langle \hat u_+(\xi-\theta_a),  
  	Q_\pm(\xi) \hat u_+(\xi-\theta_a) \rangle. \]
Changing variables by $v= \mp p^\prime(\xi)$ gives 
\[ \langle u_+, e^{i \xi  \hat v_0} u_+ \rangle
	= \int_{[-|a|,|a|]} e^{i \xi v} w(v) f_K(v;|a|) dv, \]
where		
\begin{equation}
\label{defw}
w(v) = \frac{1}{2}\sum_{m=0,1}
			 \left\{ P_{-, m}(\xi_{-,m}(v))
			 + P_{+, m}(\xi_{+,m}(v)) \right\}
\end{equation}
with $\xi_{\pm,m}(v) = m\pi + \arcsin \frac{\mp (-1)^m |b| v}{|a|\sqrt{1-v^2}}$.

\section*{Acknowledgments}  
M.M. was supported by the JSPS KAKENHI Grant Number JP15K17568.
H.S. was supported by JSPS KAKENHI Grant Number JP17K05311.
E.S. acknowledges financial support from 
the Grant-in-Aid for Young Scientists (B) and of Scientific Research (B) Japan Society for the Promotion of Science (Grant No.~16K17637, No.~16K03939).
A. S. was supported by JSPS KAKENHI Grant Number JP26800054. 
K.S acknowledges JSPS the Grant-in-Aid for Scientific Research (C) 26400156.

\end{document}